\documentclass {article}
\usepackage{fullpage}
\usepackage{authblk}

\def\BibTeX{{\rm B\kern-.05em{\sc i\kern-.025em b}\kern-.08emT\kern-.1667em\lower.7ex\hbox{E}\kern-.125emX}}

\usepackage[export]{adjustbox}
\usepackage{subfig}
\usepackage{bbm}
\usepackage{amsmath}
\usepackage{amsthm}
\usepackage{amssymb}
\usepackage{url}
\usepackage{graphicx}
\usepackage{multirow}
\usepackage{algorithm,algorithmic}
\usepackage{caption}
\usepackage{tikz}
\usetikzlibrary{calc}
\usetikzlibrary{patterns}
\tikzstyle{mybox} = [draw=black, fill=white,  thick,
rectangle, inner sep=10pt, inner ysep=20pt]
\tikzstyle{mybox} = [draw=black, fill=white,  thick,
rectangle, inner sep=2pt, inner ysep=2pt]

\usepackage{pgfkeys}

\newcommand{\declare}[1]{%
	\pgfkeys{
		/variables/#1.is family,
		/variables/#1.unknown/.style = {\pgfkeyscurrentpath/\pgfkeyscurrentname/.initial = ##1}
	}%
}

\declare{}

\newdimen\R
\newdimen\RR
\newtheorem{thm}{Theorem}
\newtheorem{lemma}{Lemma}

\newtheorem{prop}{Proposition}
\theoremstyle{definition}
\newtheorem{definition}{Definition}
\newtheorem{remark}{Remark}

\newcommand\stp{0}
\newcommand\scc{7}
\begin{document}
	
	\title{Spherical Triangle Algorithm: A Fast Oracle for Convex Hull Membership Queries }

\author{ Bahman Kalantari }

\affil{\footnotesize Department of Computer Science, Rutgers University, NJ}
\affil{\textit {kalantari@cs.rutgers.edu}}
\author{Yikai Zhang}
\affil{\textit {zhangyikai91@gmail.com}}
\date{}

\maketitle
	\begin{abstract}
		Convex Hull Membership (CHM) is the problem that inquires whether
		$p \in conv(S)$, where $p$ and the $n$ points of $S$ lie in $\mathbb{ R}^m$. Solving
		CHM alone or as part of a query problem finds applications in LP,
		CG, ML, Statistics, Topic Modeling, Minimum Volume Ellipsoid, and
		Data Reduction. For the purpose of solving CHM, Triangle Algorithm
		(TA) computes $p' \in conv(S)$ where $ p'$ is either an $\varepsilon$-approximate
		solution or a witness inducing a hyperplane separating $p$ and $conv(S)$.
		First, we prove the equivalence of the exact and approximate versions
		of CHM and Spherical-CHM, the latter the case of CHM where $p = 0$,
		$\|v\| = 1$, for all $v \in S$. We then prove that Spherical-TA, i.e.,
		TA for Spherical-CHM, terminates in $O(1/\varepsilon^2)$ iterations. Each
		iteration takes $O(mn)$ time, however with a pre-processing it could be reduced to $O(n+m)$ \cite{awasthi2018robust}.  
		 We also prove that if for each $p'$
		in $conv(S)$ with $\|p'\| > \varepsilon$ that is not a witness there exists
		$v \in S$ with $\|p' - v\| \geq \sqrt{1+\epsilon}$, then the number of
		iterations in which Spherical-TA terminates is reduced to $O(1/\varepsilon)$. In particular, it results
		in AVTA$+$, where AVTA an algorithm based on TA for
		computing all vertices of $conv(S)$. We have performed substantial
		computations on a variety of the problems mentioned above that
		indicates the TA and the Spherical-TA as efficient tools for convex hull membership query in high dimensions.

	\end{abstract}
	%
	%

	%
	
	%
	\maketitle

	\section{Introduction} \label{sec1}
	
	Given a set $S= \{v_1, \dots, v_n\} \subset \mathbb{R} ^m$ and a distinguished point $p \in \mathbb{R} ^m$, {\it Convex Hull Membership}  (CHM) 
	is the problem that inquires whether $p$ lies in  $conv(S)$, the convex hull of $S$.  CHM is a basic and fundamental problem in linear programming, computational geometry, machine learning, statistics and more. The homogeneous case of CHM, when $p=0$ arises in some fundamental polynomial time algorithms for linear programming.
	For instance, Karmarkar's algorithm  \cite{karmarkar1984new} deals with a homogeneous CHM. Another example is Khachiyan's ellipsoid algorithm \cite{khachiyan1980polynomial} which is actually designed to test the feasibility of a strict system of $n \times m$ inequalities, $Ax <b$. Using classical LP  dualities, it is easy to show the dual to the strict LP feasibility is the homogeneous CHM corresponding to the equations $A^Ty=0$, $b^Ty+s=0$. This implies homogeneous CHM is an inherent dual to strict LP feasibility. In fact homogeneous CHM admits a {\it matrix scaling} duality that leads to a simple polynomial time interior method, see \cite{khachiyan1992diagonal}. An important application of CHM in computational geometry and in machine learning is the {\it irredundancy problem}, the problem of computing all the vertices of  $conv(S)$, see e.g., \cite{toth2017handbook}.
	
	When the number of points, $n$, and dimension, $m$, are large, polynomial time algorithms for CHM are prohibitive.  For this reason {\it fully polynomial time approximation schemes} for CHM have been studied, see e.g., ~\cite{kalantari2015characterization, awasthi2018robust, gartner2009coresets, clarkson2010coresets}. These algorithms produce $\varepsilon$-approximate solution in time complexity
	 such as $mn/\varepsilon^2$, see e.g., \cite{kalantari2015characterization, clarkson2010coresets}. There are other criteria for iterative algorithms for large-scale problems,  e.g., the representation of an approximate solution and the sparsity of this representation. In CHM  an approximate solution to be represented in terms of a small number of points in $S$ is preferred. One of the well known algorithms for computing the distance from $p$ to $conv(S)$,
	sometimes known as the  {\it polytope distance} problem, is the Frank-Wolfe method  \cite{frank1956algorithm} and its variations. Letting $A$ denote the matrix  $[v_1, \dots, v_n]$ of points in $S$, $e \in \mathbb{R}^n$ the vector of ones, the Frank-Wolfe method considers the convex minimization problem:
	$\min \{f(x)=\Vert Ax - p \Vert^2 : x \in \Sigma_n\}$, where  $\Sigma_n = \{x \in \mathbb{R}^n:  e^Tx=1, x \geq 0\}$, the $n-1$ dimensional simplex. Given $x' \in \Sigma_n$, the Frank-Wolfe algorithm computes an index $j$ for which the partial derivative ${\partial f(x')}/{\partial x_j}$ is minimized. It then computes  the minimizer $x''$ of $f(x)$ along the line segment connecting $x'$ and $e_j$, one of the basis. It replaces $x'$ with $x''$ and repeats. If $x_* \in \Sigma_n$ is the optimal solution of the convex minimization,  an $\varepsilon$-approximate solution is an $x \in \Sigma_n$ such that $f(x)-f(x_*) =O(\varepsilon)$.
	The notion of {\it coreset} is related  both to representation of the approximate solutions, as well as the number of iterations of an algorithm.  The Frank-Wolfe algorithm gives an $\varepsilon$-approximate solution with $\varepsilon$-coreset of size $O(1/\varepsilon^2)$.  Clarkson \cite{clarkson2010coresets} argues that with a more sophisticated version of the algorithm that uses the  {\it Wolfe dual}, together with more computation, a coreset of size $1/\varepsilon$ can be found. Additionally, a popular class of algorithms that has $O(1/\varepsilon)$ number of iterations are the so-called {\it first-order} methods,  see the {\it fast-gradient} method of Nesterov \cite{nesterov2005smooth}.
	More generally, for the polytope distance problem, one is interested in computing the distance between two convex hulls. Gilbert's algorithm \cite{gilbert1966iterative} for the polytope distance problem coincides with the Frank-Wolfe algorithm, see G{\"a}rtner and Jaggi \cite{gartner2009coresets}.  A related problem is the {\it hard margin} support vector machine (SVM): testing if the convex hull of two finite sets of points intersect and if not, computing the optimal pair of supporting hyperplanes separating the convex hulls, see \cite{burges1998tutorial}.
	
	The {\it Triangle Algorithm} (TA), introduced in \cite{kalantari2015characterization}, is a geometrically inspired algorithm designed to solve CHM. When $p \in conv(S)$, it works analogously to the Frank Wolfe algorithm; however, the iterates are not necessarily the same and it offers more flexibility and geometric intuition.
	When $p \notin conv(S)$, the TA  computes a {\it witness}, a point $p'$ in $conv(S)$, where the  orthogonal bisector hyperplane to the line segment $pp'$ separates $p$ and $conv(S)$. This is an important feature of the TA and has proved to be very useful in several applications. As an example in \cite{awasthi2018robust}, the TA is used efficiently in \emph{All Vertex Triangle Algorithm} (AVTA) which is an algorithm for computing the set of all vertices of $conv(S)$, or an approximate subset of vertices whose convex hull approximates $conv(S)$.  The practicality and advantages of the TA over the Frank-Wolfe  are supported by large-scale computations in realistic applications. To test if $p \not \in conv(S)$,
	there is no need to compute the minimum of $f(x)$ over $\Sigma_n$.  In fact a witness $p'$ gives an estimate of the distance from $p$ to $conv(S)$ to within a factor of two. The TA in $O(1/\varepsilon^2)$ iterations computes a point $p_\varepsilon \in conv(S)$ so that either $\Vert p - p_\varepsilon \Vert \leq \varepsilon R$, where $R= \max \{\Vert p- v_i \Vert: v_i \in S \}$, or $p_\varepsilon$ is a witness.   In each iteration the algorithm uses at most one more of the $v_i$'s to represent the current approximation $p'$. It can thus be seen that when $p \in conv(S)$, the algorithm produces an $\varepsilon$-coreset of size $O(1/\varepsilon^2)$.   The complexity of the TA  improves if $p$ is contained in a ball of radius $\rho$, contained in the relative interior of $conv(S)$. Specifically,  the number of iterations to compute an $\varepsilon$-approximate solution $p_\varepsilon$ is $O((R^2/\rho^2) \log (1/\varepsilon))$.   The generalization of the TA for computing the distance between two arbitrary compact convex sets is developed in \cite{kalantari2019algorithmic}. The algorithm described in \cite{kalantari2019algorithmic} either computes an approximate point of intersection,  a separating  hyperplane, an optimal supporting pair of hyperplanes, or the distance between the sets, whichever is preferred.  The complexity of each iteration is dependent on the nature and description of the underlying sets. In the worst case, one needs to solving an LP over one or the other convex set.
	
	There are three major contributions of the current work. First, we propose a novel algorithm called the \textit{Spherical Triangle Algorithm} (Spherical-TA) and report a novel analysis on its complexity. Second, we list applications of the  Spherical-TA. In particular, we introduce two classes of problems: feasibility problems  and the irredundancy problem. Third, we provide solid computational results to verify the efficiency of the TA and the Spherical-TA in both feasibility and irredundancy problems. We also show that, as efficient oracle, the TA and the Spherical-TA can significantly impact various domains.

	The article is organized as follows: we first review the TA in Section 2. In Section 3, we prove the equivalence of exact and approximate CHM and Spherical-CHM. In Section 4, we give an $O(1/\varepsilon^2)$ iteration of the TA for Spherical-CHM. In Section 5, we prove that if in Spherical-CHM for each $p' \in conv(S)$ with $\Vert p' \Vert> \varepsilon$ that is not a witness, there exists  $v \in S$ satisfying  $\Vert p' - v \Vert \geq \sqrt{1+ \varepsilon}$,
	then the number of iterations of TA reduces to  $O(1/\varepsilon)$, matching Nesterov's fast-gradient algorithm. This geometric assumption is reasonable and suggests a strategy for when it is not satisfied at an iterate. In Section 6 and 7, as an  application of the TA, we solve the feasibility problems, i.e. strict LP feasibility and LP feasibility. In Section 8, we introduce  the irredundancy problem. In Section 9, we demonstrate our empirical results. Lastly, we conclude with remarks and propose future work.

	\section{A Summary of Triangle Algorithm, Dualities and Complexity} \label{sec2}

	The TA described in \cite{kalantari2015characterization} is an iterative algorithm for solving the CHM problem. Formally,  given a set $S=\{v_1, \dots, v_n\} \subset \mathbb{R}^m$, a distinguished point $p \in \mathbb{R} ^m$, and $\varepsilon \in (0,1)$, solving CHM means either computing an $\varepsilon$-{\it approximate solution}, i.e. $p_\varepsilon \in conv(S)$ so that
	\begin{equation}
	\Vert p - p_\varepsilon \Vert  \leq \varepsilon R,  \quad R = \max \{\Vert v_i - p \Vert  : v_i \in S\},
	\end{equation}
	or a hyperplane that separates $p$ from $conv(S)$.  Given an {\it iterate} $p' \in conv(S)$, the TA searches for a {\it pivot} to get closer to $p$: $v \in S$  is a $p$-{\it pivot} (or simply {\it pivot}) if $\Vert p' - v \Vert \geq \Vert p -  v \Vert$. Equivalently,
	\begin{equation} \label{eq2}
	(p-p')^Tv \geq  \frac{1}{2} (\Vert p \Vert^2 -  \Vert p' \Vert^2).
	\end{equation}
	
	A $p$-{\it witness} (or simply {\it witness})  is a point $p' \in conv(S)$,  where the orthogonal bisecting hyperplane to $pp'$ separates $p$ from $conv(S)$.  Equivalently,
	\begin{equation} \label{eq3}
	\Vert p' - v_i \Vert  < \Vert p - v_i \Vert, \quad \forall i=1, \dots, n.
	\end{equation}
	The separating hyperplane $H$ is given as
	\begin{equation}
	H=\{x: (p-p')^Tx = \frac{1}{2}(\Vert p \Vert^2 - \Vert p' \Vert^2)\}.
	\end{equation}
	
	Given an iterate $p' \in conv(S)$ that is neither an $\varepsilon$-approximate solution nor a witness, the TA finds a $p$-pivot $v \in S$. Then on the line segment $p'v$ it computes the closest point to $p$, denoted by $Nearest(p; p'v)$. It then replaces $p'$ with $Nearest(p; p'v)$ and repeats.
	
	\begin{prop} \label{prop1} \rm{\cite{kalantari2015characterization}}
		Suppose $p' \in conv(S)$ satisfies $\Vert p' - p \Vert \leq \min \{\Vert p - v_i \Vert: i=1, \dots, n\}$, and $v_j$ is a $p$-pivot, then the new iterate is
		\begin{equation} \label{eq5}
		p'' =Nearest(p;p'v_j)=
		(1-\alpha)p' + \alpha v_j, \quad  \alpha = {(p-p')^T(v_j-p')}/{\Vert v_j - p' \Vert^2}.
		\end{equation}
		If $p'=\sum_{i=1}^n \alpha_i v_i$, a convex combination, $p''=\sum_{i=1}^n \alpha'_i v_i$, $\alpha'_j=(1-\alpha)\alpha_j+\alpha$,  $\alpha'_i= (1-\alpha)\alpha_i$,  $\forall i \not =j$. \qed
	\end{prop}
	
	The correctness and complexity of the TA are stated in the following:
	
	\begin{thm}  \label{thm1} {{\rm (Distance Duality)\cite{kalantari2015characterization}}} 
		$p \in conv(S)$ if and only if for each  $p' \in conv(S)$ there exists a pivot $v_j \in S$. Equivalently,  $p \not \in conv(S)$ if and only if there exists a witness $p' \in conv(S)$.  \qed
	\end{thm}

	\begin{thm}   \label{thm3}  {{\rm (Complexity Bounds) \cite{kalantari2015characterization}}} 
		Given $\varepsilon \in (0,1)$, if the TA starts with $p_0$, the $v_i$ closest to $p$ , in $O(1/\varepsilon^2)$ iterations it either computes
		$p_\varepsilon \in conv(S)$ with $\Vert p - p_\varepsilon \Vert \leq \varepsilon R$, or a witness. \qed
	\end{thm}
	
	\begin{definition} \label{strictpivot} Given $p' \in conv(S)$,  $v \in S$ is a {\it strict} $p$-pivot (or simply {\it strict} pivot) if $\angle p'pv \geq \pi/2$.
	\end{definition}

	\begin{thm}  \label{thm4} {\rm (Strict Distance Duality) \cite{kalantari2015characterization}}  Assume $p \not \in S$.  Then $p \in conv(S)$ if and only if for each  $p' \in conv(S)$ there exists strict $p$-pivot $v \in S$.  \qed
	\end{thm}
	
	\begin{thm}  \label{thm5} \rm{\cite{kalantari2015characterization}}
		 Suppose $B_\rho (p)= \{x:  \Vert x - p \Vert \leq \rho R\} \subseteq conv^\circ(S)$, the relative interior of $conv(S)$. If the TA uses a strict pivot in each iteration,  $p_\varepsilon \in conv(S)$ can be computed in $O\big (\rho^{-2} \log \frac{1}{\varepsilon} \big )$ iterations. \qed
	\end{thm}


	\begin{thm}  \label{thm6} \rm{\cite{awasthi2018robust}}  Let $\widehat S= \{ \widehat v_1, \dots, \widehat v_N\}$ be a subset  of $S=\{v_1, \dots, v_n\}$. Given $p \in \mathbb{R}^m$, consider testing if $p \in conv( \widehat S)$. 
		Given $\varepsilon \in (0, 1)$, the complexity of testing if there exists an $\varepsilon$-approximate solution is
		\begin{equation} \label{eqAA}
		O \bigg (m N^2+ \frac{N}{\varepsilon^2} \bigg).
		\end{equation}

		In particular, suppose in testing if $p \in conv(S)$, $S=\{v_1, \dots, v_n\}$, the TA computes an $\varepsilon$-approximate solution $p_\varepsilon$ by examining only the elements of a subset $\widehat S= \{ \widehat v_1, \dots, \widehat v_N\}$ of $S$.  Then the  number of operations to determine if  there exists an $\varepsilon$-approximate solution  $p_\varepsilon \in conv(S)$, is  as stated in (\ref{eqAA}). $\qed$
	\end{thm}
	
	\begin{remark}
		Without any pre-processing, the straight forward iterative complexity the TA is $O(mN/\varepsilon^2)$. However, with an $O(mN^2)$ pre-processing, the complexity of each iteration is $O(N)$, resulting in the overall complexity in (\ref{eqAA}).
	\end{remark}

	\section{Spherical-CHM  and Equivalence to CHM}

	The {\it Spherical-CHM} is the case of CHM, where $p=0$ and each $v_i \in S$ has unit norm.  Given a  raw data set $S^r=\{v_1^r,...v^r_n\}$ and $p^r$, we set $p=0$ and set $S=\{v_1,...,v_n\}$, where $v_i=(v^r_i-p^r)/\|v^r_i-p^r\|$. This step scales every point onto a unit sphere. (See Figure  1)

	\begin{figure}
		
		\subfloat[convex hull of raw data]
		{
				\begin{tikzpicture}[scale=0.40]\label{unscal}
				\filldraw  (0.5*\scc,0.5*\scc) node[right] {$v^r_1$} circle (2pt);
				\filldraw (-0.5*\scc,0.5*\scc) node[left]  {$v^r_2$}circle (2pt);
				\filldraw (-0.9*\scc,-0.9*\scc) node[left]  {$v^r_3$}circle (2pt);
				\filldraw  (0*\scc,-1.2*\scc) node[below]   {$v^r_4$}circle (2pt);
				\filldraw  (0,0) node[right] {$p^r$}circle (2pt);
				\filldraw(0.7*\scc,-0.7*\scc) node[right]   {$v^r_5$}circle (2pt);
				\draw[thick,dash dot]  (0.5*\scc,0.5*\scc) --(-0.5*\scc,0.5*\scc);
				\draw[thick,dash dot]  (-0.5*\scc,0.5*\scc) --(-0.9*\scc,-0.9*\scc) ;
				\draw[thick,dash dot]   (-0.9*\scc,-0.9*\scc) --(0,-1.2*\scc);
				\draw[thick,dash dot]  (0,-1.2*\scc)--(0.7*\scc,-0.7*\scc);
				\draw[thick,dash dot]  (0.7*\scc,-0.7*\scc)--(0.5*\scc,0.5*\scc);
				
				\end{tikzpicture}
		}
		\subfloat[convex hull of scaled data]
		{
				\begin{tikzpicture}[scale=0.3]\label{unscal}
				\begin{scope}[black]
				\draw (0.0,0.0) circle (1.414*\scc);
				\end{scope}
				\filldraw  (0.5*\scc,0.5*\scc) node[right] {$v^r_1$} circle (2pt);
				\filldraw  (1*\scc,1*\scc) node[right] {$v_1$} circle (2pt);
				\filldraw (-0.5*\scc,0.5*\scc) node[left]  {$v^r_2$}circle (2pt);
				\filldraw (-1*\scc,1*\scc) node[left]  {$v_2$}circle (2pt);
				\filldraw (-1*\scc,-1*\scc) node[left]  {$v_3$}circle (2pt);
				\filldraw (-0.9*\scc,-0.9*\scc) node[right]  {$v^r_3$}circle (2pt);
				
				\filldraw  (0,-1.2*\scc) node[right]   {$v^r_4$}circle (2pt);
				\filldraw  (0,-1.414*\scc) node[below]   {$v_4$}circle (2pt);
				\filldraw  (0,0) node[right] {$p$}circle (2pt);
				\filldraw(0.7*\scc,-0.7*\scc) node[right]   {$v^r_5$}circle (2pt);
				\filldraw(1*\scc,-1*\scc) node[right]   {$v_5$}circle (2pt);
				\draw (0,0)--(1*\scc,1*\scc);
				\draw (0,0)--(-1*\scc,1*\scc);
				\draw (0,0)--(0,-1.414*\scc);
				\draw (0,0)--(1*\scc,-1*\scc);
				\draw (0,0)--(-1*\scc,-1*\scc);
				\draw[thick,dash dot]  (1*\scc,1*\scc) --(-1*\scc,1*\scc);
				\draw[thick,dash dot]  (-1*\scc,1*\scc) --(-1*\scc,-1*\scc);
				\draw[thick,dash dot]  (-1*\scc,-1*\scc) --(0,-1.414*\scc);
				\draw[thick,dash dot]   (0,-1.414*\scc)--(1*\scc,-1*\scc);
				\draw[thick,dash dot]  (1*\scc,-1*\scc)--(1*\scc,1*\scc);

				\end{tikzpicture}

		}
		
		\caption{Compute $S=\{v_1,...,v_n\}$ by scaling raw data set $S^r=\{v_1^r-p^r,...,v^r_n-p^r\}$ onto unit shpere}
		\label{fig:Fig2}
	\end{figure}

	Intuitively we expect CHM and Spherical-CHM to be equivalent. However, we need to make this precise, that is we need to convert approximate solutions and separating hyperplanes from one problem to the other.  The theorem below shows that given an instance of CHM we can convert it to an instance of Spherical-CHM so that
	the convex hull of points in CHM contains $p^r$ if and only if the convex hull of points in Spherical-CHM contains the origin. Next, it proves if we have an $\varepsilon$-approximate solution of Spherical-CHM, we can convert it to an $\varepsilon$-approximate solution of CHM. Finally, given a separating hyperplane for Spherical-CHM, we can construct a separating hyperplane for the CHM.
	
	\begin{thm}  Given $p \in \mathbb{R}^m$, $S=\{v_i: i=1, \dots, n\} \subset  \mathbb{R}^m$, $p \not \in S$, let $R= \max\{\Vert v_i-p \Vert : v_i \in S\}$.  Let $\overline p =0$, and  $\overline S= \{ \overline v_i = v_i -p: i=1, \dots, n\}$. Let $S_0 = \{ \overline v_i/ \Vert \overline v_i \Vert: i=1, \dots, n\}$.
		
		(i) {\rm (Equivalence of Exact Feasibility in CHM and Spherical-CHM)}
		
		$p \in conv(S)$ if and only if $0 \in conv(\overline S)$ if and only if $0 \in conv(S_0)$.
		
		(ii) {\rm (Equivalence of Approximate Solutions in CHM and Spherical-CHM)}
		
		Given  $\varepsilon \in (0,1)$, suppose $\widehat p_\varepsilon = \sum_{i=1}^n \alpha_i \overline v_i/ \Vert \overline v_i \Vert$, $\sum_{i=1}^n \alpha_i=1$, $\alpha_i \geq 0$ satisfies
		\begin{equation} \label{eqpep}
		\Vert \widehat p_\varepsilon \Vert \leq \varepsilon.
		\end{equation}
		Set
		\begin{equation} \label{eqbet}
		p_\varepsilon =  \sum_{i=1}^n \beta_i v_i, \quad
		\beta_i=  \frac{{\alpha_i}/{\Vert \overline v_i \Vert}}{ \sum_{j=1}^n  ({\alpha_j}/{\Vert \overline v_j \Vert})}, \quad i=1, \dots, n.
		\end{equation}
		Then
		\begin{equation}
		\Vert p - p_\varepsilon \Vert \leq \varepsilon R.
		\end{equation}
		
		(iii) {\rm (Equivalence of Separation in CHM and Spherical-CHM)}
		
		Assume $p \not \in conv(S)$. Without loss of generality assume $p=0$, hence $R= \max \{\Vert v_i \Vert: i=1, \dots, n\}$.  Let $S_R=\{v_i/R: i=1, \dots, n\}$.
		Suppose $p' \in conv(S_0)$ is a $0$-witness, i.e. the orthogonal bisector hyperplane to the line segment $0p'$, say $H_0$,  separates $0$ from $conv(S_0)$.  Let $w_i=v_i/R$, $i=1, \dots, n$. Then all $w_i$'s lie in the same hemisphere as the one enclosing $S_0$.  For each $i$, let $w_i'$ be the projection of $w_i$ onto the line segment $0p'$.  Let the closest of the $w_i'$ to the origin be denoted by $\widehat w '$. Then the orthogonal bisector hyperplane to  the line segment  $0 \widehat w'$, say $H$,  separates $0$ from $conv(S_R)$ (see Figure 2 (a)). Equivalently, a scaled version of $H$ separates $0$ from $conv(S)$.
	\end{thm}
	
	\begin{proof}  (i):  Suppose $p =\sum_{i=1}^n \alpha_i v_i$, $\sum_{i=1}^n \alpha_i =1$, $\alpha_i \geq 0$. Writing  $p=\sum_{i=1}^n \alpha_i p$, we get
		\begin{equation} \label{eqs0}
		0 =  \sum_{i=1}^n \alpha_i (v_i -p) = \sum_{i=1}^n \alpha_i \overline v_i \in conv( \overline S) .
		\end{equation}
		Since $p \not = v_i$, $\overline v_i \not =0$.  We can thus rewrite the equation in (\ref{eqs0}) as
		\begin{equation}
		\sum_{i=1}^n \alpha_i\Vert \overline v_i \Vert  \frac{\overline v_i}{\Vert \overline v_i \Vert}=0.
		\end{equation}
		Dividing both sides by $\sum_{j=1}^n \alpha_j \Vert \overline v_j \Vert$, we get $0 \in conv(S_0)$. We have thus proved one direction of the implications in (i). The other direction follows analogously.
		
		(ii):   Multiplying (\ref{eqpep}) by $R$ we get
		\begin{equation} \label{eqwork}
		\Vert \sum_{i=1}^n \frac{ \alpha_i R}{\Vert \overline v_i \Vert}  \overline v_i \Vert  \leq R \varepsilon.
		\end{equation}
		Dividing each side  of (\ref{eqwork}) by  $\sum_{j=1}^n \alpha_j R/\Vert \overline v_j \Vert$, and from the definition of the $\beta_i$'s in (\ref{eqbet}) we get,
		\begin{equation} \label{eq12}
		\Vert \sum_{i=1}^n \beta_i \overline v_i \Vert = \Vert p - \sum_{i=1}^n \beta_i v_i \Vert  \leq \varepsilon R / \sum_{i=1}^n \frac{ \alpha_i R} { \Vert \overline v_i \Vert}.
		\end{equation}
		From the definition of $R$, $R/\Vert \overline v_i \Vert \geq 1$ so that we have
		\begin{equation} \label{eq13}
		\sum_{j=1}^n  \frac {\alpha_j R} {\Vert \overline v_j \Vert} \geq \sum_{j=1}^n \alpha_j =1.
		\end{equation}
		Using (\ref{eq13}) in (\ref{eq12}), the proof of (ii) follows.
		
		(iii):  Since $p'$ is a $0$-witness, the hyperplane $H_0 = \{x: p'^T x = 0.5 \Vert p' \Vert \}$  separates $0$ from $conv(S_0)$. Thus one of the two  hemisphere whose base is parallel to $H_0$ contains all of $S_0$.
		While $H_0$ may not separate $0$ from $conv(S_R)$, the hemisphere that contains $S_0$ must also contain $S_R$. Thus the projection of $w_i=v_i/R$ onto the line segment $0p'$ and its extension to a line, strictly  lies in the hemisphere containing $S_R$. Then the projection $w_i$ that is closest to the origin gives rise to a  separating hyperplane $H$ (see Figure 2).
	\end{proof}

	\begin{figure}
		\subfloat[$conv(S_R)$ and witness $p'_R$]
		{
				\begin{tikzpicture}[scale=0.60]\label{unscal}
				

				
				
				
				\draw (1,-3+\stp) node[right] {$ \frac{v_2}{R}$};
				\draw (-2.3,-1.5+\stp) node[left] {$\frac{v_4}{R}$};
				\draw (1,-3+\stp)--(-4.94975, -4.94975+\stp);
				
				\draw (-2.3,-1.5+\stp)--(-4.94975, -4.94975+\stp) node[pos=1.0, below] {$\frac{v_3}{R}$};
				\filldraw (-4.94975, -4.94975+\stp)circle (2pt);
				\draw (-5., -.6+\stp) --(5.,-.6+\stp) node[pos=.1, above] {$H$};
				\draw (2.5,-0.5+\stp) node[below right=0.1cm] {$\frac{v_1}{R}$};
				\draw [dash dot](2.5,-1+\stp)--(0,-1+\stp)node[left] {$\widehat{w'}$};
				\filldraw (0,-1+\stp)circle (2pt);
				\draw (-2.3,-1.5+\stp)--(2.5,-1+\stp);
				\filldraw (2.5,-1+\stp)circle (2pt);
				\draw (2.5,-1+\stp)--(1,-3+\stp);
				\filldraw (1,-3+\stp) circle (2pt);
				\filldraw (-2.3,-1.5+\stp) circle (2pt);
				\filldraw  (0, -4+\stp) circle (2pt);
				\draw (0,-4+\stp) node[below] {$p'$};
				\draw (0, -4+\stp)--(0,+\stp);
				\draw (0,0+\stp) node[above] {$0$};
				\filldraw (0,0+\stp) circle (2pt);

				
				\end{tikzpicture}
		}
		\subfloat[$conv(S_0)$ and witness $p'$]
		{
				\begin{tikzpicture}[scale=0.45]\label{scal}

				
				\begin{scope}[black]
				\draw (0.0,0.0) circle (7.0);
				\end{scope}
				
				\draw (-6.65, -2) --(6.65,-2) node[pos=.07, above] {$H_0$};
				\draw (0,0)--(1,-3) node[right] {$ \frac{v_2}{R}$};
				\draw (0,0)--(-2.3,-1.2) node[left] {$\frac{v_4}{R}$};
				\draw (0,0)--(2.07,-1) node[right] {$\frac{v_1}{R}$};
				\filldraw (2.07,-1) circle (2pt);
				\filldraw (1,-3) circle (2pt);
				\draw (0,0) -- (2.21359, -6.64078) node[pos=1.0, below] {$\frac{v_2}{\Vert v_2 \Vert}$};
				\filldraw  (2.21359, -6.64078) circle (2pt);
				\filldraw (-2.3,-1.2) circle (2pt);
				\draw (0,0) -- (-6.20609, -3.23796) node[pos=1.0, left] {$\frac{v_4}{\Vert v_4 \Vert}$};
				\filldraw  (-6.20609, -3.23796) circle (2pt);
				
				\draw (0,0)--(6.26099, -3.1305) node[pos=1.0, right] {$\frac{v_1}{\Vert v_1 \Vert}$};
				\draw (0,0)-- (-4.94975, -4.94975) node[pos=1.0, below] {$\frac{v_3}{R}=\frac{v_3}{\Vert v_3 \Vert}$};
				\filldraw (-4.94975, -4.94975) circle (2pt);
				\filldraw (6.26099, -3.1305) circle (2pt);
				\draw (0,0) node[above] {$0$};
				\draw (0,-4) node[below] {$p'$};
				\draw (0.0,0.0)-- (0,-4) ;
				\filldraw (0,0) circle (2pt);
				\filldraw (0,-4) circle (2pt);
				\filldraw (0,-2) circle (2pt);

				\end{tikzpicture}

		}
		
		\begin{center}
			\caption{ $S=\{v_i: i=1,2,3,4\}$ (not drawn), $S_0=\{v_i/\Vert v_i \Vert, i=1,2,3,4\}$,
				$R= \Vert v_3 \Vert$, $S_R=\{v_i/R, i=1,2,3,4\}$.
				The point $p' \in conv(S_0)$ is a witness. In Figure \ref{unscal}, all vertices are scaled by a constant(the maximum distance between query point and vertices). The orthogonal bisecting hyperplane of $0 \widehat{w'}$, $H$, separates $0$ from $conv(S_R)$.  In Figure \ref{scal}, all vertices are scaled onto a unit sphere.The orthogonal bisecting hyperplane of $0p'$, $H_0$, separates $0$ from $conv(S_0)$. $S_R$ lies on the same hemisphere as $S_0$. }
		\end{center}
		\label{fig:Fig1}
	\end{figure}

	\section{Spherical Triangle Algorithm and its Complexity} \label{Sec4}
	
	Recall that we define the  Spherical-TA by converting a CHM into a Spherical-CHM and applying the TA.	From now on we consider CHM where $p=0$ and  $S= \{v_i: i =1, \dots, n\} \subset \mathbb{R}^m$, where $\Vert v_i \Vert =1$, for all $i=1, \dots, n$, thus a Spherical-CHM.    Consider the TA for Spherical-CHM:
	%
	%
	%
	%
	%
	%

	\subsection{Algorithm Description}
	\begin{algorithm}[H]
		\caption{Spherical-TA ($S^r=\{v_1^r,...,v^r_n\}$, $p^r$, $\varepsilon \in (0,1)$)}
		\begin{algorithmic}[1]
			\scriptsize
			\STATE	{\bf Step 0.} Compute $S=\{v_1,...,v_n\}$ where $v_i=(v^r_i-p^r)/\|v^r_i-p^r\|$. Set $p=0$.
			\STATE 	{\bf Step 1.}  Set $p'= v_1$.
			\STATE  {\bf Step 2.} If $\Vert p'\Vert \leq \varepsilon$, then output $p'$ as an $\varepsilon$-approximate solution, stop.
			\STATE 	{\bf Step 3.}  If there is a strict pivot $v_j \in S$, set $p' \leftarrow Nearest(0; p'v_j )$. Goto Step 2.
			\STATE 	{\bf Step 4.} Output $p'$ as a witness. Stop.	
		\end{algorithmic}
	\end{algorithm}

	In what follows we will derive the  worst-case complexity of Spherical-TA. The worst scenario occurs when in each iteration the iterate is not a witness, and the
	pivot is orthogonal  to the iterate (See Figure \ref{SecondCircle}).  Thus it suffices to analyze the complexity under the worst-case for each iteration.  These are formalized next and then used in the next section.
	
	\begin{lemma} \label{lem1} Given $p' \in conv(S)$,  let $v \in S$ be a strict pivot (see Figure \ref{FirstCircle}).   Let $p''= Nearest(0, p'v)$,
		$\delta=\Vert p' \Vert$,  $\delta' = \Vert p'' \Vert$, $\mu = \Vert p' - p'' \Vert$.   Let $\widehat v$ be a point of unit distance, orthogonal to $p'$ (drawn for convenience on  Figure  \ref{SecondCircle}).  Let $\widehat p''= Nearest (0, p' \widehat v)$,  $\widehat \delta' =\Vert \widehat p'' \Vert$, $\widehat \mu= \Vert p' - \widehat p'' \Vert$.   Then we have,
		
		\begin{equation} \label{eq1lem}
		\delta'^2 \leq  \frac{\delta^2}{1+ \delta^2}, \quad \widehat \delta'^2 = \frac{\delta^2}{1+ \delta^2}, \quad
		\widehat \mu ^2 =  \frac{\delta^4}{1+ \delta^2}  \geq \frac{\delta^4}{2}.
		\end{equation}
		In particular,
		
		\begin{equation} \label{eq3lem}
		\delta' \leq \widehat \delta', \quad \mu \geq \widehat \mu  \geq  \frac{\delta^2}{\sqrt{2}}.
		\end{equation}
	\end{lemma}

	\begin{proof}   By definition of strict pivot, the angle  $\angle p'0v$ is  which implies $\delta'^2/\delta^2 \leq \|ov\|^2/\|vp'\|^2$. We have $\|vp'\|^2 \geq \|op'\|^2+\|ov\|^2=\|op'\|^2+1$ which implies the first inequality  in (\ref{eq1lem}). The equality in (\ref{eq1lem}) holds because $\angle  p'0 \hat{v}$ is a right angle.
		From the similarity of the triangles $\triangle p' \widehat v 0$ and $\triangle 0 \widehat p'' \widehat v$ in Figure \ref{SecondCircle} we may write ${\widehat \mu}/{\delta}= {\widehat \delta'}/{1}$. Squaring and substituting for $\widehat \delta'^2$, we get the expression for $\widehat \mu^2$ in (\ref{eq1lem}). The lower bound is obvious.  The first and last inequalities in (\ref{eq3lem}) follow from (\ref{eq1lem}). The second inequality follows from (\ref{eq1lem}) and from,
		\begin{equation*}
		\mu^2=  \delta^2 - {\delta'}^2 \geq \delta^2 -  \frac{\delta^2}{1+ \delta^2} = \widehat \mu^2. \qedhere
		\end{equation*}
	\end{proof}

	\begin{thm} \label{thm7}  
		For $k \geq 0$, let $\delta_k = \Vert p_k \Vert$, where $p_k$ is the sequence of iterates of the TA, $p_0=v_1$ and none of the iterates is a witness. Let $\widehat{\delta}_0 = \delta_0$ and define
		\begin{equation} \label{rec}
		\widehat \delta^2_{k+1}= \frac{\widehat \delta^2_k}{1+ \widehat \delta^2_k}, \quad k \geq 0.
		\end{equation}
		Then for all $k \geq 1$,
		\begin{equation}
		\delta_k \leq \widehat \delta_k.
		\end{equation}
	\end{thm}
	
	
	\begin{proof} We prove this by induction on $k$. From Lemma \ref{lem1} the inequality is true for $k=1$.  Assume true for $k$.  The function $g(t)= t/(1+t)$ is monotonically increasing  on $(0, \infty)$.  From the relationship between $\delta_{k+1}$ and $\delta_k$ in Lemma \ref{lem1}, together with monotonicity of $g(t)$, we may write
		\begin{equation}
		\delta^2_{k+1} \leq \frac{\delta^2_k}{1+\delta^2_k} \leq \frac{\widehat \delta^2_k}{1+ \widehat \delta^2_k} = \widehat \delta^2_{k+1}.
		\end{equation}
	\end{proof}

	\begin{figure}[htpb]
		\centering
		\begin{tikzpicture}[scale=0.45]

		
		\begin{scope}[black]
		\draw (0.0,0.0) circle (7.0);
		\end{scope}
		
		\draw (0.0,0.0) -- (.7*8, .7*6) -- (0,-4.0) -- cycle;
		\filldraw (.7*8, .7*6) circle (2pt);
		\filldraw (.33*.7*8, .33*.7*6)+(0,-.67*4) circle (2pt);
		\draw (0,-4) -- (0,0) node[pos=0.5, left] {$\delta$};
		\draw (0,0) -- (1.848, -1.294) node[pos=0.5, below] {$\delta'$};
		\draw (0,0) node[left] {$o$};
		\draw (2.8, 2.1) node[above] {$1$};
		\draw (5.6, 4.2) node[right] {$v$};
		\draw (0,-4) node[below] {$p'$};
		\draw (1.848, -1.294) node[right] {$p''$};
		\draw (0,-4) -- (1.848, -1.294) node[pos=0.5, right] {$\mu$};
		\filldraw (0,0) circle (2pt);
		\filldraw (0,-4) circle (2pt);

		\end{tikzpicture}
		\begin{center}
			\caption{An iteration of the TA at $p'$ with a strict pivot $v$: $p''$ is projection of $0$ on $p'v$, $\delta=\Vert p' \Vert$, $\delta'=\Vert p'' \Vert $,  $\mu= \Vert p' - p'' \Vert$.} \label{FirstCircle}
		\end{center}
	\end{figure}

	\begin{figure}[htpb]
		\centering
		\begin{tikzpicture}[scale=0.45]

		
		\begin{scope}[black]
		\draw (0.0,0.0) circle (7.0);
		\end{scope}
		
		\draw (0,0) -- (7,0);
		\draw (0,-4) -- (0,0) node[pos=0.5, left] {$\delta$};
		\draw (0,-4) -- (7, 0);
		\draw (0,0) node[left] {$o$};
		\draw (7,0) node[right] {$\widehat v$};
		\draw (0,-4) node[below] {$p'$};
		\draw (1.722, -3.016)  node[below] {$\widehat p''$};
		\draw (0,0) -- (1.722, -3.016)  node[pos=0.5, right] {$\widehat \delta'$};
		\filldraw (1.722, -3.016) circle (2pt);
		\draw (0,-4) -- (1.722, -3.016)   node[pos=0.5, above] {$ \widehat \mu$};
		\filldraw (0,0) circle (2pt);
		\filldraw (7,0) circle (2pt);
		\filldraw (0,-4) circle (2pt);

		\end{tikzpicture}
		\begin{center}
			\caption{An iteration of triangle algorithm at $p'$ with least reduction if a strict pivot $\widehat v$ is orthogonal to $p'$: $\widehat p''$ projection of $0$ on $p'\widehat v$, $\delta=\Vert p' \Vert$, $\widehat \delta'=\Vert \widehat p'' \Vert $,  $\widehat \mu= \Vert p'-\widehat p'' \Vert$.} \label{SecondCircle}
		\end{center}
	\end{figure}
	
	\begin{thm} \label{thm8}
		  Consider Spherical-CHM. The TA  terminates in $O(1/\varepsilon^2)$ iterations with $p_\varepsilon \in conv(S)$, either a witness or $\Vert p_\varepsilon \Vert \leq \varepsilon $.
	\end{thm}
	\begin{proof}  Let $p_k$, $\delta_k$ and $\widehat \delta_k$  be as in the previous theorem. We claim for any natural number $N$,
		\begin{equation}
		\widehat \delta^2_N=  \frac{1}{1+N}.
		\end{equation}
		This is true for $N=1$. By the induction hypothesis and the recursive definition of $\widehat \delta_i$, in (\ref{rec}), we have,
		\begin{equation}
		\widehat \delta^2_{N+1} = \frac{1}{2+N}=\frac{1}{1+(N+1)}.
		\end{equation}
		In particular, if $N = \lceil {1}/{\varepsilon}\rceil$, we get
		\begin{equation} \label{eqkey1}
		\widehat \delta_N =  \frac{1}{\sqrt{1+N}} \leq  \frac{1}{\sqrt{1+ {1}/{\varepsilon}}}  =  \frac{ \sqrt{\varepsilon}}{\sqrt{1+ \varepsilon}} \leq  \sqrt{\varepsilon}.
		\end{equation}
		From Theorem \ref{thm7} $\delta_k \leq \widehat \delta_k$ for all $k \geq 1$. From this and (\ref{eqkey1}) if $0 \in conv(S)$,  in $O(1/\varepsilon)$ iterations TA computes $p_k$ such that $\delta_k \leq \sqrt{\varepsilon}$.  To complete the proof it suffices to replace $\sqrt{\varepsilon}$ with $\varepsilon$.
	\end{proof}
	
	\section{Improved complexity analysis for Spherical-TA}
	
	\begin{definition}  Given a Spherical-CHM, we say a point $p' \in conv(S)$ that is not a witness  and for which $\Vert p' \Vert > \varepsilon$, has the $\varepsilon$-property if   there exists is a pivot $v$ such that
		\begin{equation}
		\Vert p' - v \Vert \geq \sqrt{1+ \varepsilon}.
		\end{equation}
	\end{definition}
	
	As an example if the ball of radius $\sqrt{\varepsilon}$ is contained in $conv(S)$, then  Spherical-CHM has the $\varepsilon$-property everywhere outside of the ball of radius $\varepsilon$.  We now establish an improved complexity for Spherical-TA  with the $\varepsilon$-property.
	
	\begin{thm}  Consider a Spherical-CHM. If every iterate $p' \in conv(S)$ of the TA  that is not a witness and for which $\Vert p'\Vert > \varepsilon$ has the $\varepsilon$-property, then in $O(1/\varepsilon)$ iterations, either the TA computes a witness, or $p_\varepsilon \in conv(S)$ such that $\Vert p_\varepsilon \Vert \leq \varepsilon$.
	\end{thm}
	
	\begin{proof}  
		Note $\varepsilon\leq \sqrt{\varepsilon}$.
		If $0 \in conv(S)$ from Theorem \ref{thm8}, in $O(1/\varepsilon)$ iterations we get an iterate $p_{k_0}$ such that $\Vert p_{k_0} \Vert \leq \sqrt{\varepsilon}$.  
		If $\Vert p_{k_0}\Vert  \leq \varepsilon$, we are done. 
		Otherwise let $k=k_0$ and 
		we claim that $p_{k+1}$ will decrease the gap sufficiently.
		More precisely, we claim
		\begin{equation} \label{claimx}
		\delta_{k+1}^2 \leq \delta^2_k - (\sqrt{2}-1)^2 \varepsilon^2 \leq \varepsilon - (\sqrt{2}-1)^2 \varepsilon^2.
		\end{equation}
		To prove (\ref{claimx}), on the one hand we have
		\begin{equation}
		\delta_{k+1}^2 = \delta^2_k - \mu_k^2.
		\end{equation}
		Consider Figure \ref{ThirdCicle} and assume $p'=p_k$, $\delta=\delta_k$, $\delta'=\delta_{k+1}$, $\mu= \Vert p'- p''\Vert = \mu_k$, $v=v_k \in S$ satisfying $\Vert p' - v \Vert \geq \sqrt{1 + \varepsilon}$.  Let $q$ be the point on $vp'$, where $\Vert v - q \Vert =1$.
		Note that $p''$ must be closer to $v$ than to $q$. Thus,
		\begin{equation}
		\mu_k \geq  \Vert p' - q \Vert \geq \sqrt{1+\varepsilon} - 1 \geq (\sqrt{2}-1) \varepsilon \geq 0.4 \varepsilon.
		\end{equation}
		Also, since $\delta_k \leq \sqrt{\varepsilon}$, we have proved (\ref{claimx}). Hence the number of iterations $k \geq k_0$ to get $\delta^2_{k+1} \leq \varepsilon^2$ is  $O(1/\varepsilon)+1=O(1/\varepsilon)$.
	\end{proof}
	
	\begin{figure}[htpb]
		\centering
		\begin{tikzpicture}[scale=0.45]

		
		\begin{scope}[black]
		\draw (0.0,0.0) circle (7.0);
		\end{scope}
		\draw (6, 3.6) circle (7.0);
		\draw (0.0,0.0) -- (.7*8, .7*6) -- (0,-4.0) -- cycle;
		\filldraw (5.6, 4.2) circle (2pt);
		\filldraw (.33*.7*8, .33*.7*6)+(0,-.67*4) circle (2pt);
		\draw (0,-4) -- (0,0) node[pos=0.5, left] {$\delta$};
		\draw (0,0) -- (1.848, -1.294) node[pos=0.5, below] {$\delta'$};
		\draw (0,0) node[left] {$0$};
		\draw (5.6, 4.2) node[right] {$v$};
		\draw (0,-4) node[below] {$p'$};
		\draw (1.848, -1.294) node[right] {$p''$};
		\draw (0,-4) -- (1.848, -1.294) node[pos=0.5, right] {$\mu$};
		\draw (0,-4) -- (1.848, -1.294) node[pos=0.7, right] {$q$};
		\filldraw (.81*1.848, -.19*4-.81*1.294) circle (2pt);
		
		\filldraw (0,0) circle (2pt);
		\filldraw (0,-4) circle (2pt);

		\end{tikzpicture}
		\begin{center}
			\caption{At iterate $p'$ the pivot $v$ satisfies $\Vert p'- v \Vert \geq \sqrt{1+ \varepsilon}$. The circle of radius one centered at $v$ intersect $p'v$ at $q$ and  $\mu=\Vert p'- p''\Vert \geq \Vert p' - q \Vert \geq (\sqrt{2}-1) \varepsilon$.} \label{ThirdCicle}
		\end{center}
	\end{figure}
	
	\begin{definition} We say an iterate $p_k \in conv(S)$ is $\varepsilon$-reduced at an iterate $p_t \in conv(S)$, $t >k$, if
		\begin{equation}
		\Vert p_t \Vert^2 \leq \Vert p_k \Vert^2 - (0.4 \varepsilon)^2.
		\end{equation}
	\end{definition}
	
	The strategy we propose when we get an iterate
	$p_k$ that does not have $\varepsilon$-property and is not a witness, is to compute $p_t$, if possible, that $\varepsilon$-reduces $p_k$, and in the simplest way possible. Then restart the ordinary TA with $p_t$, checking if it in turn has the $\varepsilon$-property and so on.  Suppose $v^k$ is a strict pivot for $p_k$.   We compute the nearest point to $0$ on $p_kv^k$ to get $p_{k+1}$. Next, we compute a strict pivot $v^{k+1}$ for $p_{k+1}$. Let the {\it restricted} $\varepsilon$-approximate Spherical-CHM be the problem of  testing if $0 \in conv(\{p_k,v^{k}, v^{k+1} \})$.  At each iteration in solving the restricted problem we check if the corresponding iterate, say $p_t$, $\varepsilon$-reduces $p_k$. If so, we start from $p_t$.  Otherwise, we obtain a {\it relative} witness, say $p_t$. Next, we compute a strict pivot in $S$, say $v^{k+2}$ (if $p_t$ is not a witness with respect to $S$).  We then augment the restricted Spherical-CHM to testing if $0$ is in $conv(\{p_k, v^{k}, v^{k+1}, v^{k+2} \})$ and repeat the process. This process would stop either with a witness with respect to $S$,  or an iterate $p_t$ that $\varepsilon$-reduces $p_k$ and we return to the ordinary TA with $p_t$ as the current iterate.
	
	The worst-case complexity of such a composite iterate is unknown at this time. However,  considering the geometry of the points in $S$, we would expect that this complexity  depends on $n$ and $m$ and the relationship between them.  For $n >> m$ and when the pairwise distance between points in $S$ is reasonably large,
	 one would expect that the composite iterate will stop after a few iterations so that the overall number of iterations would remain to be $O(1/\varepsilon)$. However, the complexity of a composite iteration may exceed $O(n+m)$. The theoretical analysis of the composite iterate is nevertheless an interesting open problem. 
	In order to investigate this problem it may be useful to first construct difficult problems for the Spherical-TA with a fixed $m$, where $0\in conv(S)$
	  so that most points are on one semi-sphere, $\varepsilon$-property is not satisfied for a current iterate, and $n$ is as small as possible. When increasing $n$, at some point $\varepsilon$-property will be satisfied in the next iteration. If the pairwise distance between points is $\delta>0$, then there must be a relationship between $\varepsilon$, $\delta$, $m$ and the minimum number $n$ so that the $\varepsilon$-property will be satisfied in next iteration. We feel that understand such examples are necessary in the investigation of the complexity of the  composite iterative Spherical-TA.
%
%

	\section{Solving Strict Linear Feasibility as Spherical-CHM} \label{Str_LP}
	The following lemma connects strict LP feasibility to CHM and is a consequence of Gordan's Theorem, hence also provable via Farkas Lemma:
	
	\begin{lemma} Let $A$ be an $n \times m$ real matrix and $b \in \mathbb{R}^n$.  Then $Ax < b$ is feasible if and only if there is no feasible solution to the homogeneous CHM: $A^Ty=0$, $b^Ty+s=0$, $\sum_{i=1}^n y_i+s=1$, $y \geq 0$, $s \geq 0$. \qed
	\end{lemma}
	
	The next theorem shows if we have a witness for the homogenous CHM dual of strict linear feasibility, it solves the strict linear feasibility itself. In particular, the TA can test the solvability of  strict linear feasibility.

	\begin{thm}  \label{homodd}   For $i=1, \dots, n+1$, let $v_i$ be the $i$-th column of the $(m+1) \times (n+1)$ matrix  $
		B=
		\begin{pmatrix}
		A^T&0 \\
		b^T&1\\
		\end{pmatrix}
		$. Suppose  $conv(\{v_1, \dots, v_{n+1}\})$ does not contain the origin. Let
		$p'=\begin{pmatrix} x\\
		\alpha\\ \end{pmatrix} \in \mathbb{R}^{m+1}$ be
		a witness.  Then $A(-x/\alpha) <b$.
	\end{thm}

	\begin{proof} We will use the distance duality (Theorem 1) to prove the theorem.  Denote the rows of $A$ by $a_i^T$. Then for $i=1, \dots, n$,  $v_i=\begin{pmatrix} a_i\\
		b_i\\ \end{pmatrix}$ and $v_{n+1}=\begin{pmatrix} 0\\
		1\\ \end{pmatrix}$, all in $\mathbb{R}^{m+1}$. Since $p'$ is a witness,
		\begin{equation}  \label{eqdd0}
		\Vert p' - v_i \Vert^2 < \Vert v_i \Vert^2, \quad \forall i=1, \dots, n+1.
		\end{equation}
		From (\ref{eqdd0}) we get
		\begin{equation} \label{eqdd00}
		\Vert x- a_i \Vert^2 + (\alpha - b_i)^2 <  \Vert a_i \Vert^2 + b_i^2, \quad \forall i=1, \dots, n.
		\end{equation}
		Simplifying (\ref{eqdd00}) we get
		\begin{equation} \label{eqdd1}
		-2(a_i^Tx + \alpha b_i) < - (\Vert x \Vert^2 + \alpha^2), \quad \forall i=1, \dots, n.
		\end{equation}
		From (\ref{eqdd0}) for $i=n+1$ we get,
		\begin{equation} \label{eqdd2}
		\Vert x \Vert^2 + (1-\alpha)^2 < 1.
		\end{equation}
		From (\ref{eqdd2}) $\alpha >0$.  This  gives:
		\begin{equation}
			-(\Vert x \Vert^2+ \alpha^2) <0 
		\end{equation}
		so that from (\ref{eqdd1})
		\begin{equation} \label{eqdd3}
		-a_i^Tx  <   \alpha b_i, \quad \forall i=1, \dots, n.
		\end{equation}
		Dividing both sides of (\ref{eqdd3}) by $\alpha$ implies $-x/\alpha$ is a feasible solution to the strict linear feasibility problem.
	\end{proof}
	
	\begin{remark} Without loss of generality we may assume that the first $n$ columns of  $B$ have unit norm.  Clearly the $(n+1)$-th column has unit norm. Thus the CHM corresponding to $B$ can be assumed to be Spherical.
	\end{remark}

	\section{Spherical-TA for LP feasibility}\label{LP_Fea}
		The LP feasibility problem is to test the feasibility of :\\
	\begin{equation} \label{con_l_s}
		\begin{aligned}
			&  Ax=b, \quad x\geq0.\\
		\end{aligned}
	\end{equation}
	\\
	In other words, to test if $b$ lies in the  cone of columns of $A$. i.e. $b \in cone(A)=\{ y| y=\sum_{j=1}^{n} \alpha_j A_j, \alpha_j \geq 0\}$ where $A_j, j=1,...,m$ are columns if $A$. It is well known that this problem is equivalent to the general Linear Programming problem. Given a bound $M$ on the feasible solution of \eqref{con_l_s}, it can be converted into the following  convex hull membership problem:
	
	\begin{equation} \label{lp_con}
		\begin{aligned}
			%
			\begin{pmatrix} A & 0 & -b\\ e^\top  & 1  & -M \\ 0 & 0& 1\end{pmatrix}\begin{pmatrix} \alpha \\ \beta \\ \gamma\end{pmatrix}& =\begin{pmatrix}0 \\0\\ \frac{1}{M+1}\end{pmatrix}\\
			e^\top \alpha+\beta+\gamma=1, \quad &\alpha,\beta,\gamma \geq 0.\\
		\end{aligned}
	\end{equation}
	\\
	
	where $e\in \mathbb{ R}^m$ is the vector of ones. It is easy to show that \eqref{con_l_s} is feasible iff \eqref{lp_con} is feasible. This suggests that the TA and the Spherical-TA can be applied to solve the LP feasibility problem.
	
	%
	%
	%

	\section{Spherical-TA for computing all vertices of a convex hull}
	Given a set of $n$ points $S=\{v_1,\dots,v_n\}$ in $\mathbb{ R}^m$, computing all vertices of its convex hull, known as the {\it irredundancy problem} ~\cite{toth2017handbook}, is an important problem  in computational geometry and machine learning.
	This problem becomes challenging as $m$ grows, especially for classical algorithms such as Gift Wrapping \cite{dixon1993gift} and QuickHull \cite{barber1996quickhull} due to their exponential running times in terms of the dimension. The irredundancy problem can be solved via $O(n)$ membership queries, i.e. for each point, checking  if it is an extreme point of the convex hull. One can take LP as an oracle for membership query, however, it is impractical to solve $O(n)$ LPs for large scale problems. The  All Vertex Triangle Algorithm (AVTA) algorithm  has been proposed to tackle the efficiency issue for this class of problems \cite{awasthi2018robust}. Sharing a similar  spirit with the membership type method,  AVTA applies the TA as a  membership query oracle and computes all vertices of the convex hull of a set of points under a natural assumption called $\gamma$ robustness \cite{arora2013practical} (see Definition \ref{gamma_ass} below).  Given a set of points $S=\{v_1,\dots,v_n\}$, we denote by $T$, $T \subset S$, the set of vertices of $conv(S)$.
	
	\begin{definition}\label{gamma_ass}
		The convex hull of $S$ is $\gamma$-robust if the minimum  distance from each vertex of the convex hull to the convex hull of the remaining vertices is at least $\gamma$. (See Figure \ref{Fig_gamm}a.)
	\end{definition}
	\begin{figure}[H]
		\subfloat[$\gamma$ robust convex hull]
		{
			\adjustbox{valign=c}	
			{
				\begin{tikzpicture}[scale=0.79]
				\draw (0.0,0.0) -- (4,3.0) -- (8,2.0) --(9,0)--(5,-2)-- cycle;
				\draw (0,0) node[below] {$v_1$};
				\draw (9,0) node[right] {$v_2$};
				\draw (4,3) node[above] {$v_3$};
				\draw (5,-2) node[below] {$v_4$};
				\draw (8,2) node[above] {$v_5$};
				\filldraw (5,-2) circle (2pt);
				\filldraw (0,0) circle (2pt);
				\filldraw (8,2) circle (2pt);
				\filldraw (9,0) circle (2pt);
				\filldraw (4,3) circle (2pt);
				\draw (7.9,1.7) node[below] {$\gamma$};
				\draw (9,0) -- (4,3);
				\draw(7.35,1.0)--(8,2);
				\draw(7.22,1.05)--(7.32,1.19)--(7.42,1.13);
				

				\end{tikzpicture}
			}
		}
		\subfloat[Pathological Case]
		{
			\adjustbox{valign=c}	
			{
%
				\begin{tikzpicture}[scale=1]
					\draw (0:\R)
					 \foreach \x in {20,40,...,340}
					 {
						-- (\x:\R)
						
					}
					-- cycle (360:\R);
					
					\foreach \i in {1,2,...,18}
					{
						\pgfmathtruncatemacro{\x}{ \i*20};
						\draw (\x:\RR) node {$v_{\i}$};
						\filldraw (\x:\R) circle (1pt);
					}
				\end{tikzpicture}
			}
		}
		\caption{}
		\label{Fig_gamm}
	\end{figure}
	The intuition behind the $\gamma$ robust assumption is as follows:  a vertex is important if it is far away from the convex hull of the remaining vertices. The number of vertices of a $\gamma$ robust convex hull is  much smaller compared to the number of vertices of a `non-robust' convex hull.
	For instance, consider $S$ as an $\varepsilon$-Net from a unit sphere $U$, say $\mathcal{N}_\varepsilon \subset U$ so that $\forall x\in U, \exists v\in \mathcal{N}_{\varepsilon}$,  such that $\|u-v\| \leq \varepsilon$. Every point in $\mathcal{N}_\varepsilon$ is a vertex of  $conv(\mathcal{N}_\varepsilon)$. The size of  $\mathcal{N}_\varepsilon$ could be exponential in terms of dimension. In such a  case, every vertex is of $O(\varepsilon)$ distance to the convex hull of the remaining vertices; thus,  no single vertex is important to the geometrical structure of $conv(\mathcal{N}_\varepsilon)$ (See Figure \ref{Fig_gamm}b). In such a pathological cases, instead of computing all vertices, one will need a `good' subset of vertices to approximate $conv(\mathcal{N}_\varepsilon)$. Indeed,  AVTA also works in such approximation schemes. We refer interested readers to \cite{awasthi2018robust}. In this paper, we only consider the irredundancy problem under the $\gamma$-robustness assumption.

	The property $\gamma$- robustness allows one to test whether a subset $\widehat{S}$ contains all vertices of $conv(S)$. Specifically, a set $\widehat{S} \subset S$ contains every vertex of $conv(S)$ if every point in $S$ is  within a distance less than $\gamma$ to $conv(\widehat{S})$. As an approximation, the TA can exploit such a property to solve the membership query with precision $\gamma/2R$ where $R$ is the diameter of $conv(S)$. Indeed, if the query point $p \notin conv(\widehat{S})$, the TA will return a  hyperplane $H$ which separates $p$ and $conv(\widehat{S})$ (see the distance duality in \cite{kalantari2015characterization}). This allows one to find a vertex by the following observation: \emph{ the set of 
		farthest points along the normal direction of a hyperplane always contains an extreme point.} Formally, given a hyperplane defined by its normal direction $c'$, maximizers of  $c'^T v$ over $conv(S)$ includes a vertex. In Figure \ref{Fig_Hyper}, set of vertices of $conv(S)$ is $T=\{v_1,...,v_7\}$ and the subset $\widehat{S}=\{v_1,..., v_5\}$ does not contain all vertices of $conv(S)$ as $v_6, v_7$ are excluded. Given a query point $p \notin conv(\widehat{S})$, the TA returns a witness $p'$ and $H$ the bisecting hyperplane of $pp'$. 	In Figure\ref{Fig_Hyper}a, the set of farthest points from $H$ is a single point, the vertex $v_6$.  In Figure \ref{Fig_Hyper}b, the set of farthest points above $H$ are points on the line segment $v_6v_7$ (red line). In such cases, the set of farthest points will be a facet of $conv(S)$. One can capture a missing vertex of $conv(S)$ by picking any point of $S$ on the facet and finding its farthest point on this facet.

	
	\begin{figure}[htpb]
		\subfloat[]
		{
			\adjustbox{valign=c}	
			{
				\begin{tikzpicture}[scale=0.555]
				\draw (0.0,0.0) -- (4,3.0) -- (8,2.0) --(7,0)--(5,-2)-- cycle;
				\draw (0,0) node[below] {$v_1$};
				\draw (7,0) node[right] {$v_2$};
				\draw (4,3) node[above] {$v_3$};
				\draw (5,-2) node[below] {$v_4$};
				\draw (8,2) node[above] {$v_5$};
				\filldraw (5,-2) circle (2pt);
				\filldraw (0,0) circle (2pt);
				\filldraw (8,2) circle (2pt);
				\filldraw (7,0) circle (2pt);
				\filldraw (4,3) circle (2pt);
				\filldraw (2,3) circle (2pt) node[above] {$p$};
				\filldraw (-1.1,3.3) circle (2pt) node[above] {$v_6$};
				\filldraw (-1.1,1.75) circle (2pt) node[left] {$v_7$};
				
				\filldraw (2.7,1.6) circle (2pt) node[below] {$p'$};
				\begin{scope}[black]
				\draw (-2,0.275 )-- (8,5.275);
				\end{scope}
				\draw (8,6) node[above] {$H$};
				\draw (2,3) -- (2.7,1.6);
				\draw[dash dot] (0,1.275) --(-1.1,3.3);
				\draw[dash dot] (-0.69,1.015) --(-1.1,1.75);
				\end{tikzpicture}
				
			}
		}
		\subfloat[]
		{
			\adjustbox{valign=c}	
			{
				\begin{tikzpicture}[scale=0.55]
				\draw (0.0,0.0) -- (4,3.0) -- (8,2.0) --(7,0)--(5,-2)-- cycle;
				\draw (0,0) node[below] {$v_1$};
				\draw (7,0) node[right] {$v_2$};
				\draw (4,3) node[above] {$v_3$};
				\draw (5,-2) node[below] {$v_4$};
				\draw (8,2) node[above] {$v_5$};
				\filldraw (5,-2) circle (2pt);
				\filldraw (0,0) circle (2pt);
				\filldraw (8,2) circle (2pt);
				\filldraw (7,0) circle (2pt);
				\filldraw (4,3) circle (2pt);
				\filldraw (1.2667,3.1) circle (2pt) node[above] {$v_7$};
				\filldraw (2,3) circle (2pt) node[right] {$p$};
				\filldraw (0,2.5) circle (2pt) node[above] {$v_6$};
				\begin{scope}[red]
				\draw (0,2.5)-- (1.2667,3.1) ;
				\end{scope}
				
				\filldraw (2.7,1.6) circle (2pt) node[below] {$p'$};
				\begin{scope}[black]
				\draw (-2,0.275 )-- (8,5.275);
				\end{scope}
				\draw (8,6) node[above] {$H$};
				\draw (2,3) -- (2.7,1.6);
				\draw[dash dot] (0,2.5) --(0.5,1.525);
				\draw[dash dot] (1.2667,3.1)  --(1.7667,2.125);
				\end{tikzpicture}

			}
		}

%
		
		\caption{Set of farthest points along the direction $p-p'$ contains an extreme point. In $(a)$ the set of farthest point $w.r.t$ $H$ is a single point $v_6$. In $(b)$, set of farthest point $w.r.t$ are points on the line segment $v_6v_7$ (red line).}
		\label{Fig_Hyper}
	\end{figure}

	The above approach is the motivation behind the AVTA: It iteratively adds a new vertex $v'$ to $\widehat{S}$  by computing a separating hyperplane, until all points are within distance $\gamma/2$  to $conv(\widehat{S})$. Next we give a detailed description of the AVTA.
	Given a working subset $\widehat S$ of $S$,  initialized with $v_1 \in S$ which has the maximum norm, the AVTA randomly selects $v \in S \setminus \widehat S$. It then tests via the Triangle Algorithm if $dist(v, conv(\widehat S)) \leq \gamma/2$. If so $v$ can not be a vertex thus labeled as a redundant point which will not be considered in further computation. In case $\Vert v -conv(\widehat S) \Vert > \gamma/2$, AVTA computes a $v$-witness $p' \in conv(\widehat S)$. The vector $c'=v-p'$ leads to a hyperplane  separating $v$ and $conv(\hat S)$. By maximizing $c'^Tv_i$ where $v_i$ ranges  in $S \setminus \widehat S$, one can find a new vertex $v' \notin \widehat{S}$ and the working set $\widehat{S}$ will be updated by including $v'$. If $v$ coincides with $v'$, the AVTA selects a new point in $S \setminus \widehat S$.  Otherwise, the AVTA continues to test if the same $v$ (for which a witness was found) is within a distance of $\gamma/2$ of the convex hull of the augmented set $\widehat S$. Also, as an iterate, the AVTA uses the same witness $p'$. The algorithm stops when each points in $S$ is detected  either as a redundant point or an extreme point. We next describe AVTA more precisely.

	\subsection{Algorithmic Description of AVTA}
	\begin{algorithm}
		\caption{AVTA and AVTA$+$  ($S$, $\gamma \in (0,1)$)}
		\begin{algorithmic}[1]
			\scriptsize
			\STATE 	{\bf Step 0.} Set $\widehat S = Farthest(v,S)$, where $ Farthest(v,S)$  is the set of points in $S$ that are farthest from $v$. 
			\STATE  {\bf Step 1.} Randomly select $v \in S \setminus \widehat S$.
			\STATE 	{\bf Step 2.} \textbf{Option I:} {\bf TA} $(\widehat S, v,  \gamma/(2R))$.
		 \phantom. \phantom. \phantom. \phantom.\phantom.\phantom. \phantom. \phantom. \phantom. \phantom. \phantom.\phantom.\phantom. \phantom. \phantom. \phantom. \phantom. \phantom.\phantom.\phantom. \phantom. \phantom. \phantom. \phantom. \phantom.\phantom.\phantom. \phantom. \phantom. \phantom. \phantom. \phantom.\phantom.\phantom. \phantom. \phantom. \phantom. \phantom. \phantom.\phantom.\phantom. 	
			\phantom. \phantom. \phantom. \phantom. \phantom.\phantom.\phantom. \phantom. \phantom. \phantom. \phantom. \phantom.\phantom.\phantom. \phantom. \phantom. \phantom. \phantom. \phantom.\phantom.\phantom. \phantom. \phantom.
			\phantom.\phantom.\phantom.\phantom. \phantom.
			\phantom.\phantom.\phantom.\phantom. \phantom. \phantom. \phantom. \phantom. \phantom.\phantom.\phantom. \phantom. \phantom. \phantom. \phantom. \phantom. \textbf{Option II:} {\bf Spherical-TA} $(\widehat S, v,  \gamma/(2R))$.
			\STATE 	{\bf Step 3.} If the output $p'$ of Step 2 is a $v$-witness then Goto Step 4. Otherwise, $p'$ is a $\gamma/2$-approximate solution to $v$.  Set $S \leftarrow S \setminus \{v\}$. \phantom. \phantom. \phantom. \phantom. \phantom. \phantom.\phantom. If $S= \emptyset$, stop. Otherwise, Goto Step 1.
			\STATE 	{\bf Step 4.} Let $c'=v-p'$. Compute $S'$, the set of optimal solutions of $\max \{c'^T x:  x \in S \setminus  \widehat S \}$. Randomly select $v' \in S'$. \phantom. \phantom. \phantom. \phantom. \phantom. \phantom.\phantom. \phantom. \phantom.\phantom. \phantom. \phantom.\phantom. \phantom. \phantom.\phantom. \phantom. \phantom.\phantom. \phantom. \phantom.\phantom. Replace $v'$ with an element in  $Farthest(v', S')$,  $\widehat S \leftarrow \widehat S \cup \{v'\}$
			\STATE 	{\bf Step 5.} If $v = v'$, Goto Step 1. Otherwise, Goto Step 2.
			
		\end{algorithmic}
	\end{algorithm}
	 Recall that in section \ref{Sec4} we introduced Spherical-TA as a  variant of  Triangle algorithm which can be directly applied in AVTA to replace TA as a membership query oracle. Throughout this paper, we use AVTA to represent the original version with the TA and AVTA$+$ if it applies the Spherical-TA.
	
	\subsection{Applications of AVTA}
%
	AVTA has various applications, including NMF (nonnegative matrix factorization) and Topic Modeling  which relies on the robustness of the AVTA in recovering vertices of the convex hull of a set of perturbed points. We refer the readers to \cite{awasthi2018robust} for details of such class of problems. In this paper, we focus on the \emph{Size Reduction} problem.
	Following the goal of irredundancy, the AVTA$+$ can be applied to reduce an overcomplete dataset i.e., a dataset that can be expressed by a small fraction of itself. In other words, it is applicable when given an $m\times n$ matrix $A$ as data,  the convex hull of the columns of $A$, denoted by $conv(A)$, has $K$ vertices, where $K \ll n$. In certain problems, instead of keeping the full data set which is of size $n$, one only needs to focus on the $K$ vertices. This suggests applying the AVTA$+$ as a pre-processing algorithm to remove non-extreme points in $A$.
	Such problems include, Conditional Gradient, \cite{clarkson2010coresets}, Minimum Volume Enclosing  Ellipsoid (MVEE) \cite{moshtagh2005minimum,sun2004computation} and Convex Hull Approximation \cite{blum2016sparse}. We  briefly introduce MVEE here and demonstrate the improvement of efficiency brought about by AVTA$+$.
	The MVEE estimator is based on the smallest volume ellipsoid that covers $conv(A)$. The MVEE problem has been studied for decades and has attained interest in broad areas, e.g., 
	outliers detection \cite{van2009minimum}.  Given the MVEE of a data set, one can identify outliers by picking points on the boundary \cite{van2009minimum}. Improving the efficiency of algorithms that solve the MVEE problem will impact areas such as robust statistics. Formally, the MVEE is defined as follows:\\
	\begin{equation}
	\begin{aligned}
	&\min\limits_{M} \log \text{det}(M^{-1})\\
	&(v_i-b)^TM (v_i-b)\leq 1,\quad i=1,...,n\\
	&M \succ 0.
	\end{aligned}
	\end{equation}
	where the optimization computes a vector $b \in \mathbb{ R}^m$ and an $m\times m$ symmetric and positive definite matrix $M$ given a set of points $S=\{v_1,...,v_n\}$.
	 Then the resulting ellipsoid $E =\{x| (x-b)^TM (x-b) \leq 1\}$ will be the MVEE  centered at $b$ and contains the convex hull of the columns of $A$, the data points.
	 This suggests that  one can run AVTA$+$ before solving MVEE since the number of vertices $K$ is generally much less than the number of columns of  of $A$ (See Figure \ref{fig_mvee}). 
	
	\begin{figure}[!h]
		\centering
		\subfloat[MVEE with redundant points]{
			\centering
			\includegraphics[width=0.41\textwidth]{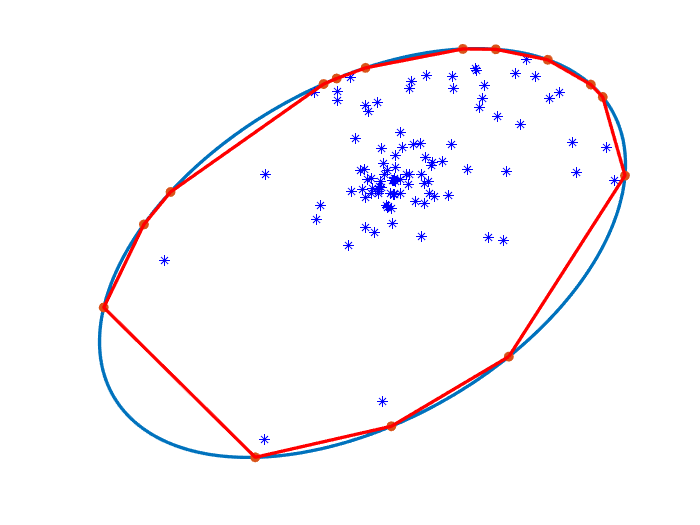}
		}
		\subfloat[MVEE with only vertices]{
			\centering
			\includegraphics[width=0.41\textwidth]{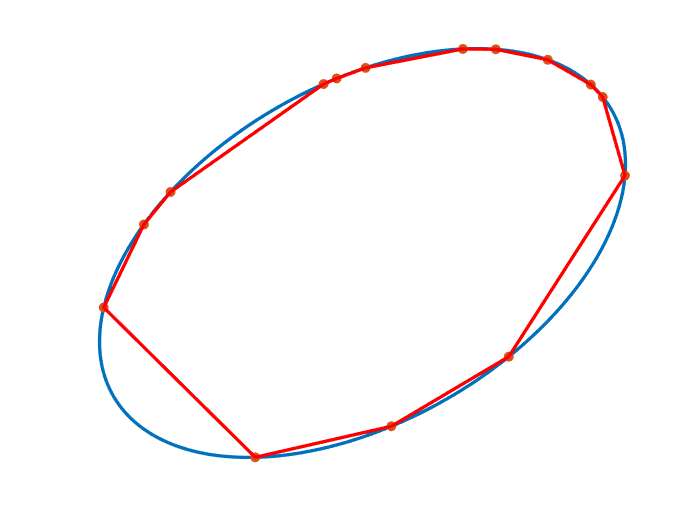}
		}
		\caption{}
		\label{fig_mvee}
	\end{figure}

	\section{Experiments }\label{Exp}\footnote{Source code: https://github.com/yikaizhang/Spherica\_TA}
	In this section we demonstrate the power of the TA and the Spherical-TA in solving CHM and the significance of CHM in solving other problems.  In Section \ref{Exp_fea} we compare the efficiency of the Spherical-TA, TA, and LP solver  for  solving CHM (\ref{Exp_chm}), LP feasibility (\ref{Exp_CLS}) and strict LP Feasibility (\ref{Exp_Str_LP}). In Section \ref{Exp_fav}  we apply the Spherical-TA as a separating hyperplane oracle to find all vertices of a convex hull of a finite set. In Section \ref{Exp_mvee}, we use the AVTA and AVTA$+$ as preprocessing steps for the MVEE problem. \\
	
	\textbf{Implementation Details:}
	We apply our implementations of the TA, the Spherical-TA, the AVTA, and the AVTA$+$ using MATLAB \footnote{Any advice or opinions posted here are our own, and in no way reflect that of MathWorks.}. In particular, we have a \emph{practical} implementation of the TA and the Spherical-TA using both the aforementioned strict pivot and also the \emph{anti-pivot} described  in ~\cite{zhangtriangle}. Our implementation of Spherical-TA also incorporates a heuristics that augmenting $S$ using random convex combination of points in $S$. For the LP solver, we use the \textit{linprog} package provided by MATLAB. For the QuickHull solver we use the \textit{convhulln} package provided by MATLAB. For the MVEE we apply  \textit{MinVolEllipse} package provided by \cite{moshtagh2005minimum}.

	
	\subsection{Feasibility:  CHM, LP Feasibility, Strict LP Feasibility}\label{Exp_fea}
	
	\subsubsection{Convex hull membership}\label{Exp_chm}

	%
	%
	\noindent In our experiments, we generate data sets in two ways. One leverages on the  Gaussian distribution, i.e. $v_i \sim \mathcal{N}(0,\mathcal{I}_m), i = {1,...,K}$ and the other on the unit sphere, i.e.,
	vertices of the convex hull are generated by  uniformly  picking points on a unit sphere. The Gaussian distribution is a natural parametric distribution widely used in statistics. The unit sphere can be viewed as a scaled version of high dimensional spherical Gaussian.
	We represent the dataset as a matrix $A \in \mathbb{R}^{m \times n}$, where $m$ is the dimension and $n$ is the number of data points. We compare the efficiency of the following three algorithms for solving  the CHM problem:  The Simplex method ~\cite{chvatal1983linear}, the TA ~\cite{kalantari2015characterization}, and the Spherical-TA. The size of the problems varies from $m=100,n=500 $ to $ m=1000,n=5000$ and the value of precision parameter Epsilon varies from $0.01 \sim 0.001$. The running times of the three	algorithms (in log scale) are shown in Figures \ref{Tab_CHM_G} and \ref{Tab_CHM_U}. One can observe that the TA and the Spherical-TA outperform other iterative algorithms. In addition, they have much better efficiency than the LP solver with a large value of precision parameter $\varepsilon$. This is because the number of iterations of the Spherical-TA to obtain an  $\varepsilon$ approximate solution increases with smaller value of $\varepsilon$. We also observe that the running time of the TA and the Spherical-TA increases linearly with $m$ and $n$ while the LP solver is more sensitive to large value of $m,n$.
	
	\begin{figure}[!h]
		\centering
		\subfloat[]{
			\centering
			\includegraphics[width=0.49\textwidth]{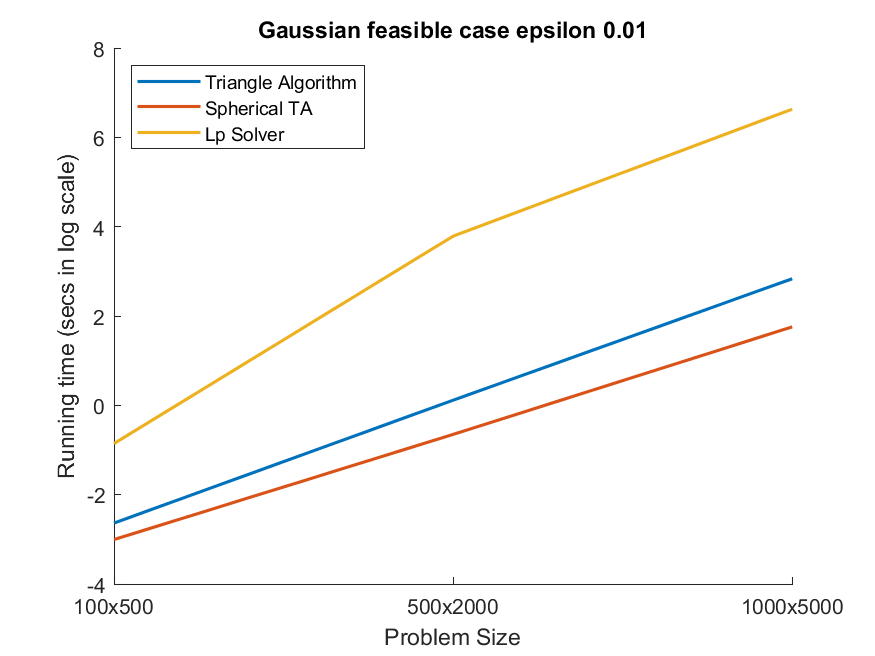}
		}
		\subfloat[]{
			\centering
			\includegraphics[width=0.49\textwidth]{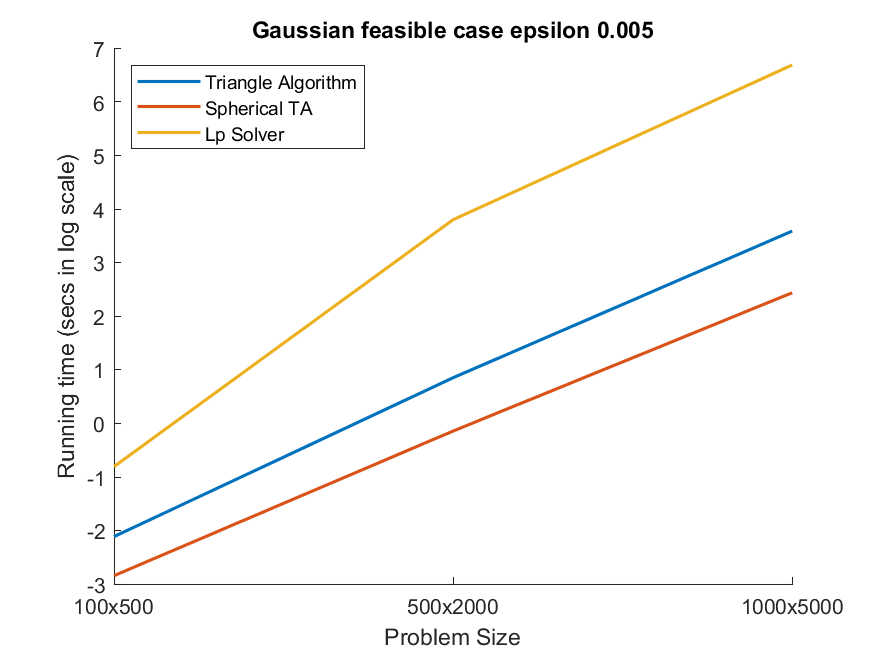}
		}
		~\\
		\subfloat[]{
			\centering
			\includegraphics[width=0.49\textwidth]{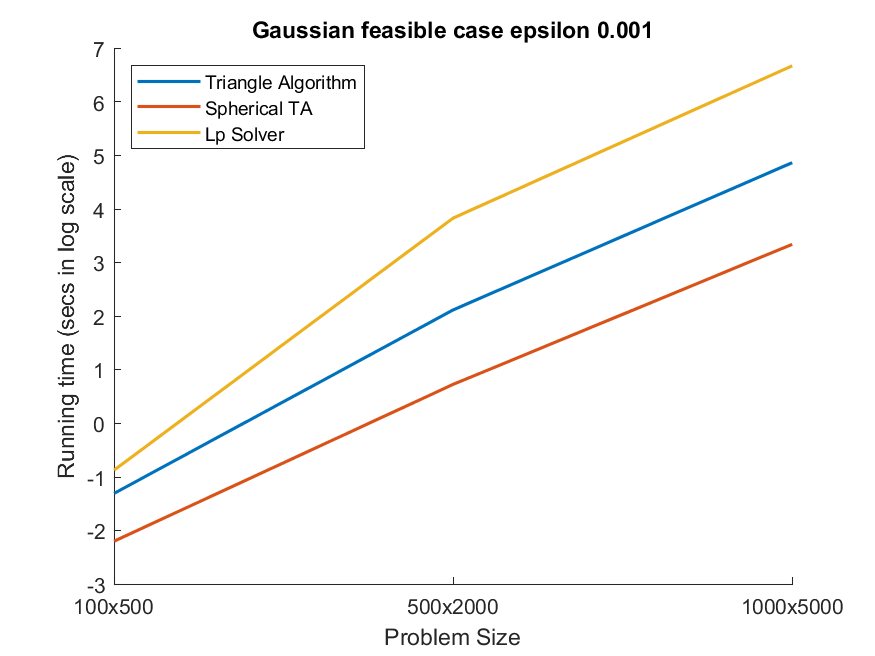}
		}
		\subfloat[]{
			\centering
			\includegraphics[width=0.49\textwidth]{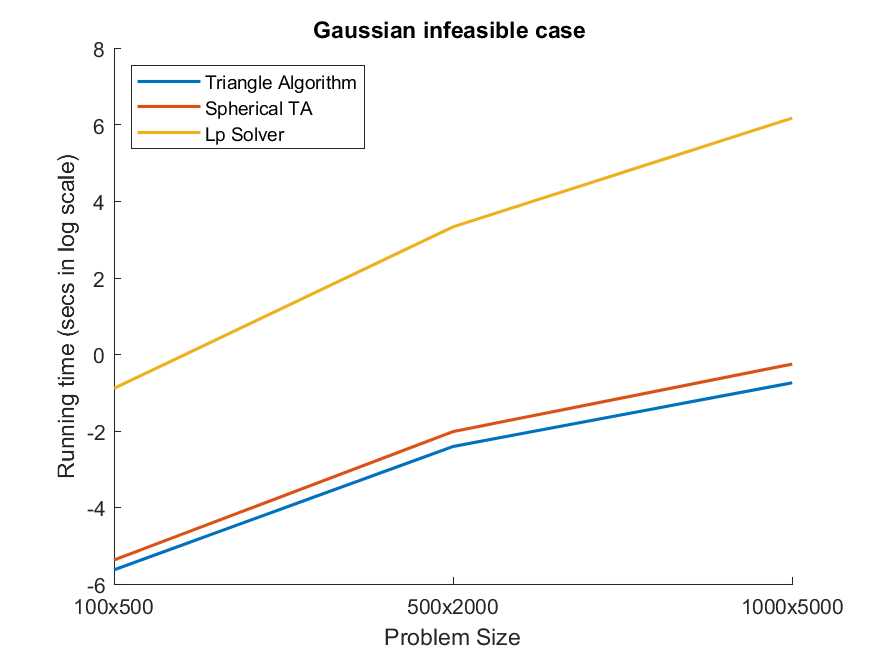}
		}
		\caption{Running time of different algorithms on CHM problem with Gaussian Vertices}
		\label{Tab_CHM_G}
	\end{figure}

	\begin{figure}[!h]
		\centering
		\subfloat[]{
			\centering
			\includegraphics[width=0.49\textwidth]{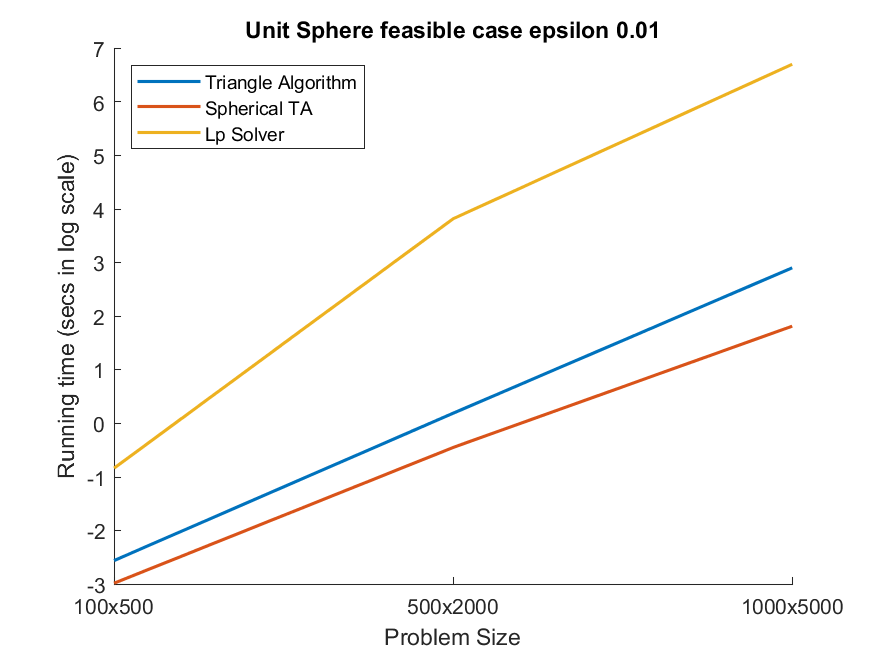}
		}
		\subfloat[]{
			\centering
			\includegraphics[width=0.49\textwidth]{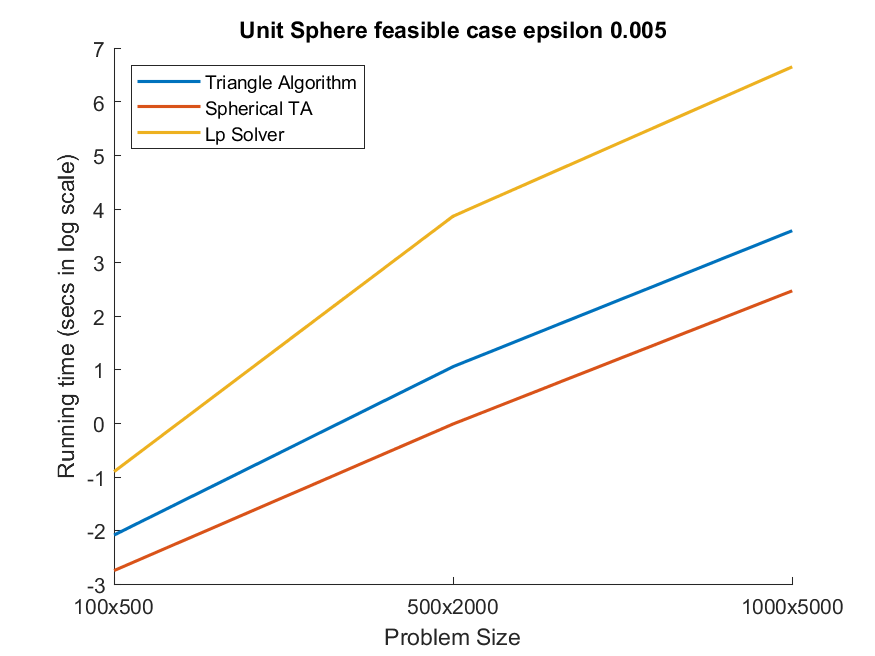}
		}
		~\\
		\subfloat[]{
			\centering
			\includegraphics[width=0.49\textwidth]{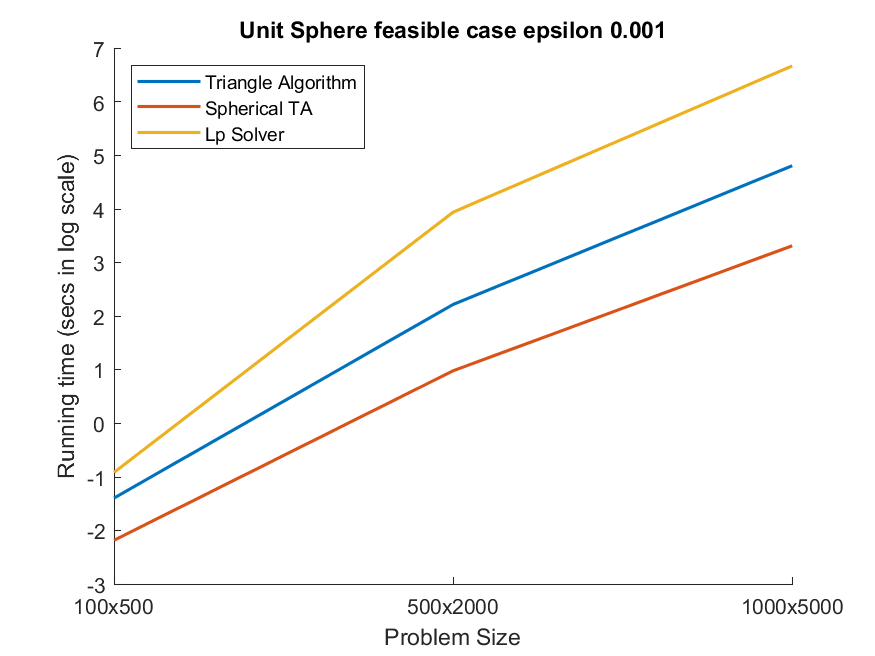}
		}
		\subfloat[]{
			\centering
			\includegraphics[width=0.41\textwidth]{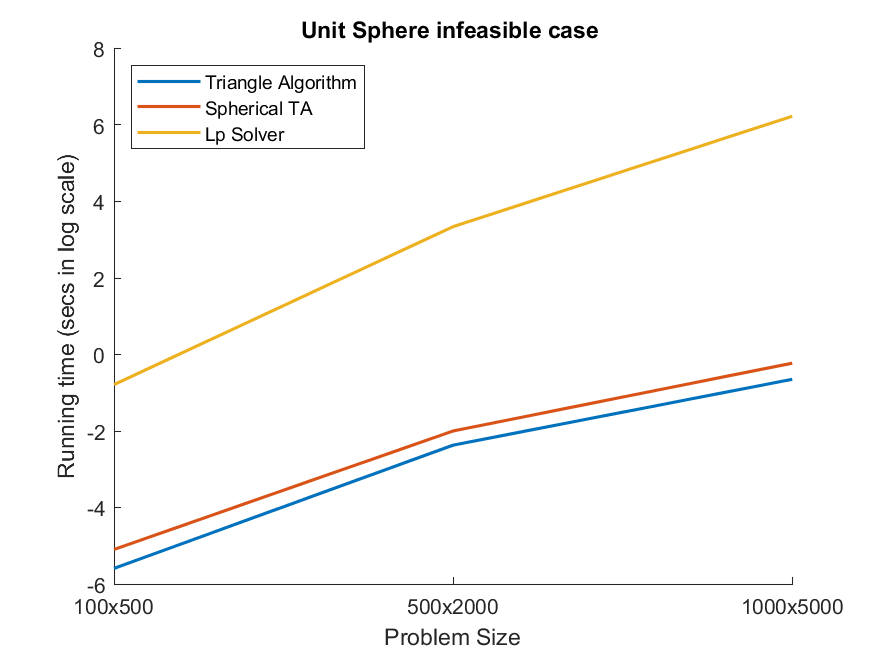}
		}
		\caption{Running time of different algorithms on CHM problem with vertices on Unit Sphere}
		\label{Tab_CHM_U}
	\end{figure}

	\subsubsection{LP feasibility} \label{Exp_CLS}
	Here we compare the efficiency of the TA , Spherical-TA and an LP solver for the LP feasibility problem introduced in Section \ref{LP_Fea}. We compare the running times of the three algorithms on datasets with different dimension, number of points, precision parameter and generator for the vertices. We generate the columns of the coefficient matrix  $A$ uniformly randomly from a unit sphere or an i.i.d Gaussian distribution. In the case $Ax=b, x\geq0$ is feasible,  we generate $x \in \mathbb{R}^m$, the solution of the linear system, as an entrywise uniform $(0,1)$ distributed vector and compute $b$ as $b=Ax$. In the case $Ax=b, x\geq0$ is infeasible, we apply an SVD: $\tilde A = U\Sigma V$, thresholding half of the singular values to be zeros and obtain  a low rank version of $ \Sigma$ denoted as $ \Sigma'$. The vector is obtained by $b= \tilde A x$ where $x$ is an entrywise normal $\mathcal{N}(0,1)$ distributed vector and $\tilde A =  U \tilde \Sigma V$ where $\tilde \Sigma$  is $\Sigma'$ perturbed by an Gaussian random matrix.
	We  set an upper bound on $x$, as $M=1200$ in all cases. The size of the problems varies from $m=50,n=500 $ to $ m=200,n=2000$ and the value of precision parameter Epsilon varies from $10^{-6} \sim 10^{-7}$. The running times of the three
	algorithms (in log scale)  are shown in Figures \ref{Tab_LP_fea_Unif} and \ref{Tab_LP_fea_U}. In particular, we observe a similar performance between the TA and the Spherical-TA in the number of iterations. Such observations suggest the complexity improvement of the  Spherical-TA over the TA is not universal. One can observe that the TA and Spherical-TA outperforms the LP solvers.
	
	\begin{figure}[!h]
		\centering
		\subfloat[]{
			\centering
			\includegraphics[width=0.41\textwidth]{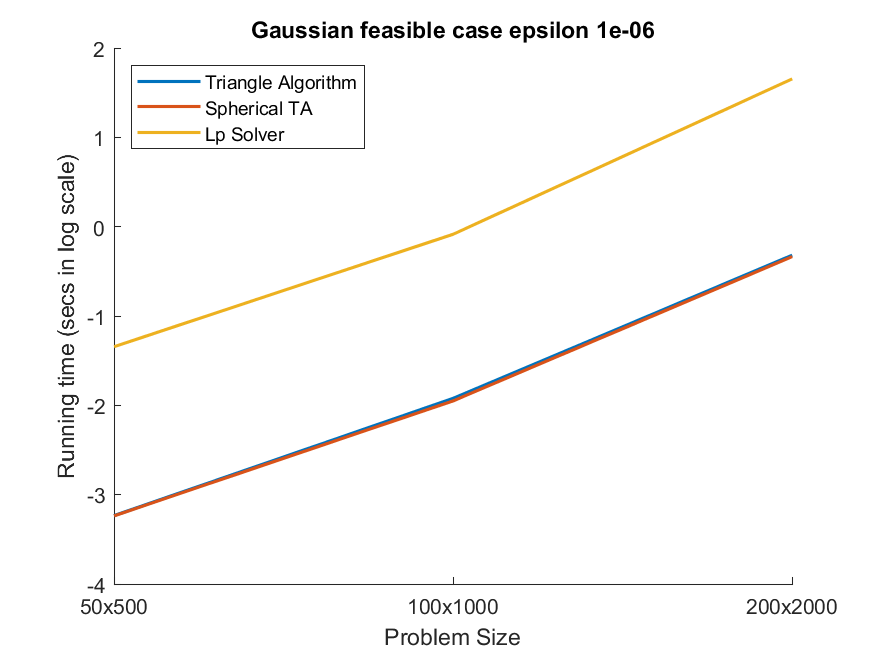}
		}
		\subfloat[]{
			\centering
			\includegraphics[width=0.41\textwidth]{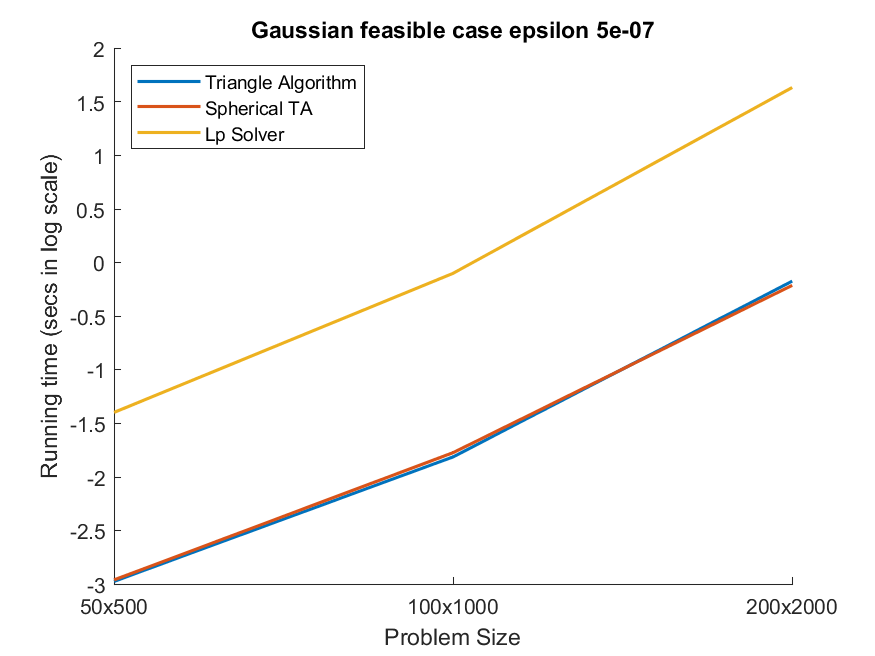}
		}
		~\\
		\subfloat[]{
			\centering
			\includegraphics[width=0.41\textwidth]{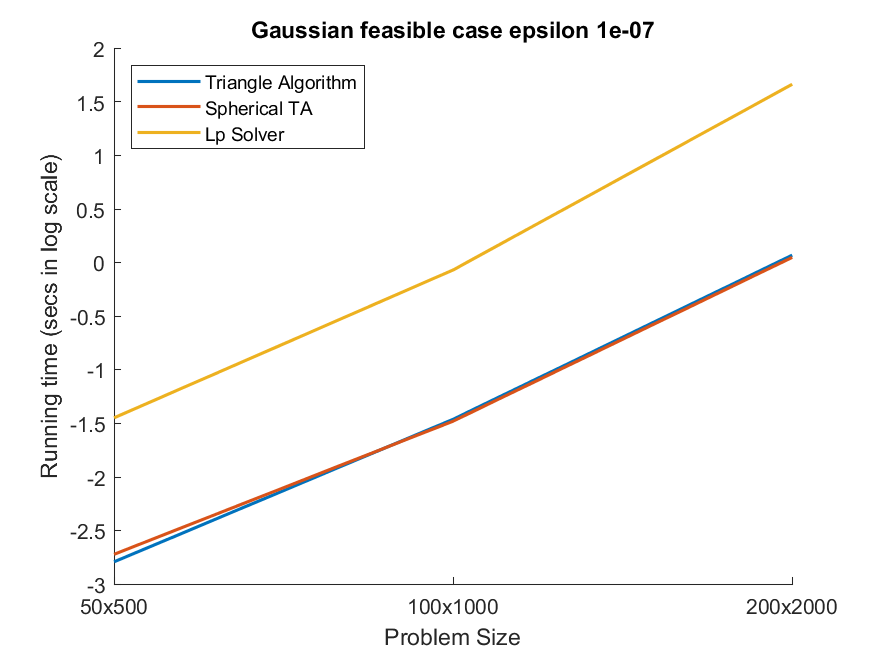}
		}
		\subfloat[]{
			\centering
			\includegraphics[width=0.41\textwidth]{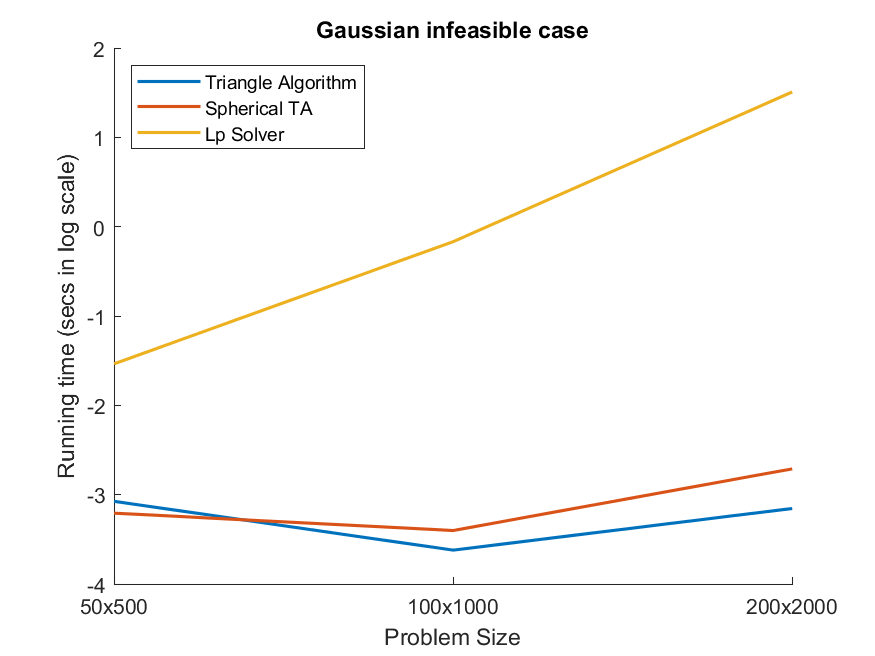}
		}
		\caption{Running time of different algorithms on LP feasibility problem (Gaussian)}
		\label{Tab_LP_fea_Unif}
	\end{figure}
	\begin{figure}[!h]
		\centering
		\subfloat[]{
			\centering
			\includegraphics[width=0.41\textwidth]{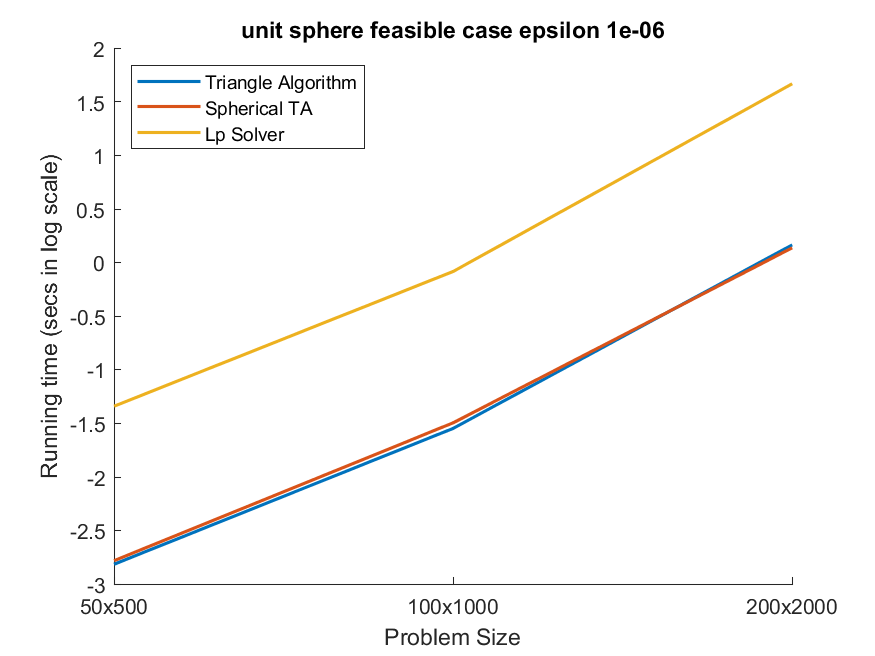}
		}
		\subfloat[]{
			\centering
			\includegraphics[width=0.41\textwidth]{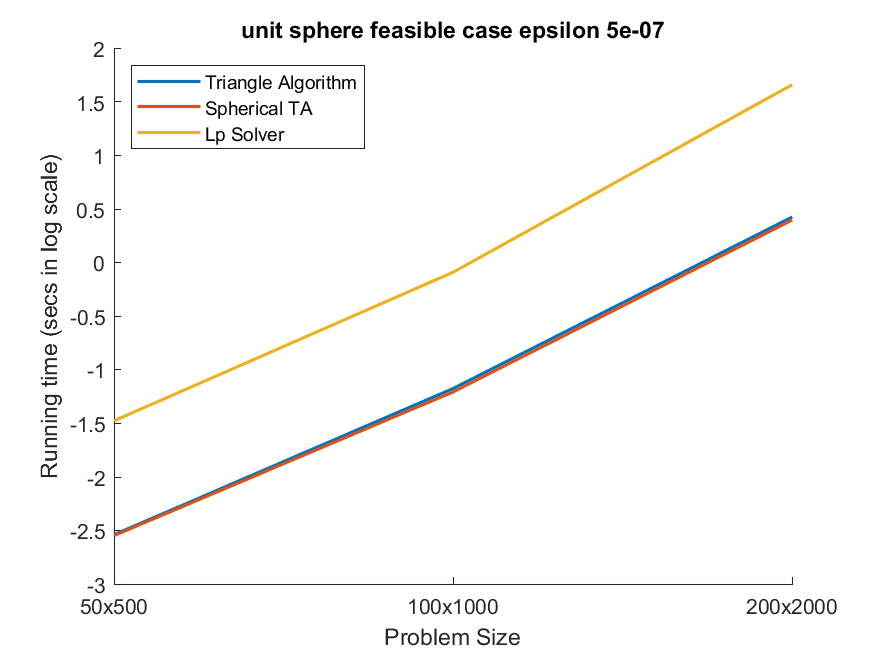}
		}
		~\\
		\subfloat[]{
			\centering
			\includegraphics[width=0.41\textwidth]{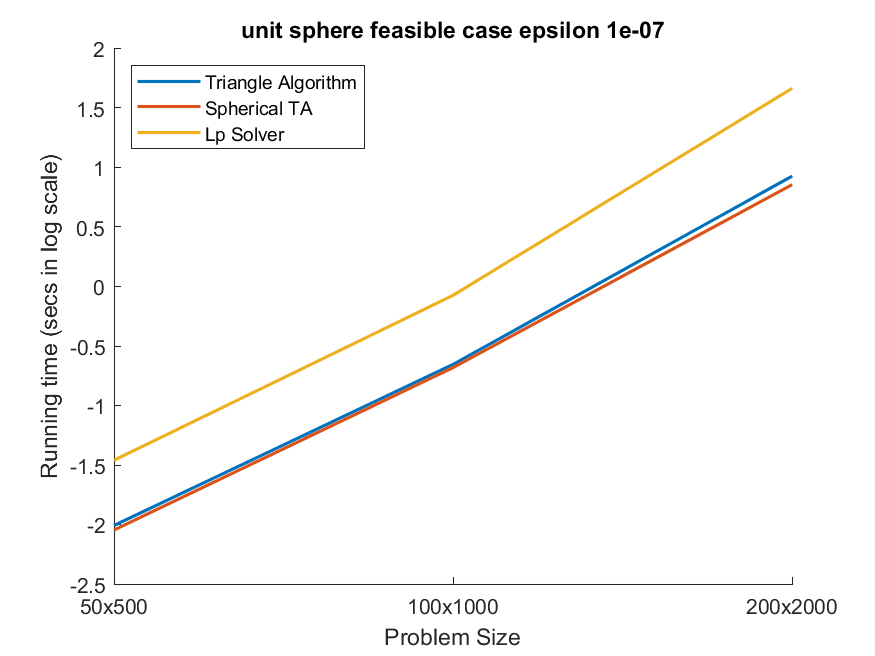}
		}
		\subfloat[]{
			\centering
			\includegraphics[width=0.41\textwidth]{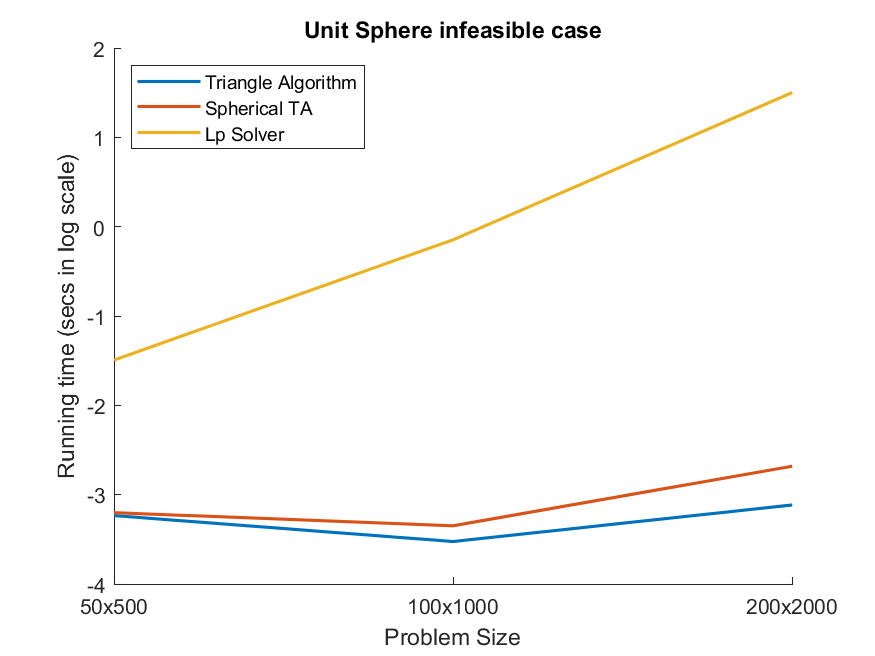}
		}
		\caption{Running time of different algorithms on LP feasibility problem (Unit Sphere)}
		\label{Tab_LP_fea_U}
	\end{figure}

	\subsubsection{Strict LP feasibility} \label{Exp_Str_LP}
	Here we solve the Strict LP feasibility problem, $Ax<b, x\geq 0$, using the TA , the  Spherical-TA and the LP solve. We compare the running times of the aforementioned three algorithms on data sets of  different sizes, feasible and infeasible cases, and different distributions for generating the coefficient matrix.  We  generate columns of $A$  $:1)$ uniformly randomly from a unit sphere; and $2)$ entrywise standard normal distributed.  For the  case of feasible $Ax<b$, $b$ is computed by $b= Ax+0.5 e $ where  $e$ is vector of ones. For the case of infeasible $Ax<b$, we generate $b= Ax-2e$ to ensure the infeasibility. The size of the input matrix $A$ varies from $n=50,m=500 $ to $ n=500,m=2500$. The running times of the three
	algorithms (in log scale)  are shown in Figures \ref{Tab_str_LP_fea_G} and \ref{Tab_str_LP_fea_Sp}.

	\begin{figure}[!h]
		\centering
		\subfloat[]{
			\centering
			\includegraphics[width=0.41\textwidth]{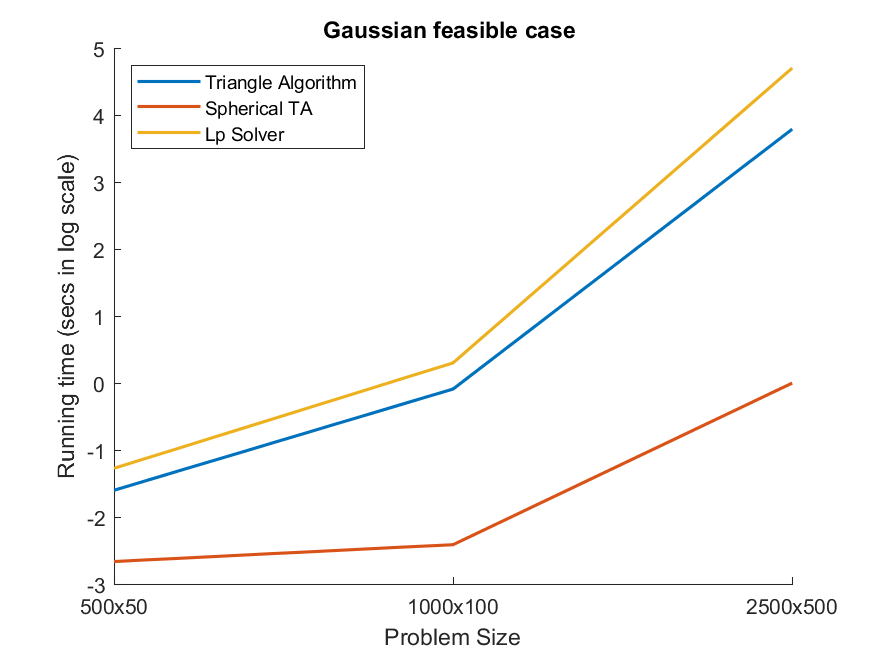}
		}
		\subfloat[]{
			\centering
			\includegraphics[width=0.41\textwidth]{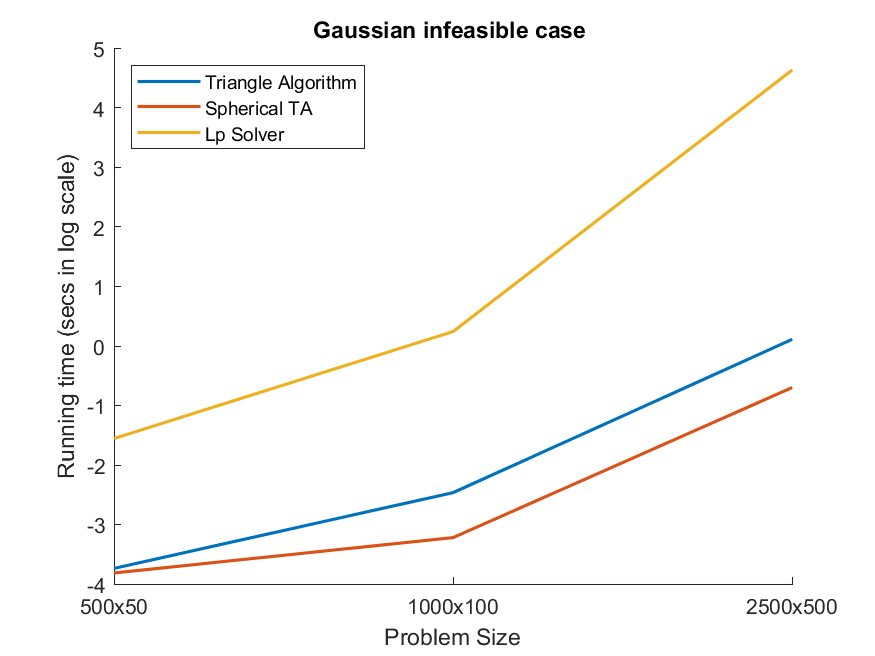}
		}
		\caption{Running time of different algorithms on Strict LP feasibility problem(Unit Sphere)}
		\label{Tab_str_LP_fea_G}
	\end{figure}
	
	\begin{figure}[!h]
		\centering
		\subfloat[]{
			\centering
			\includegraphics[width=0.41\textwidth]{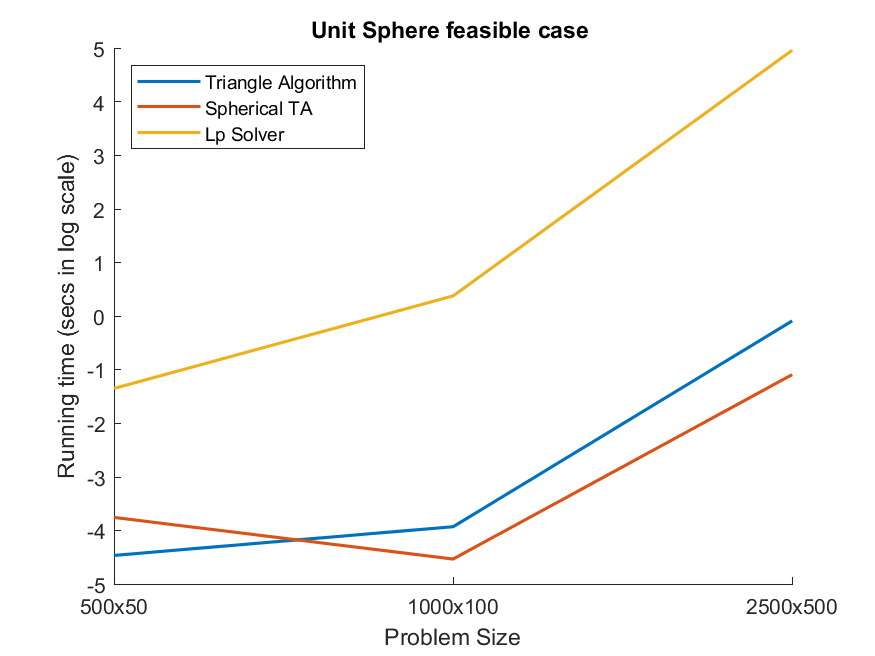}
		}
		\subfloat[]{
			\centering
			\includegraphics[width=0.41\textwidth]{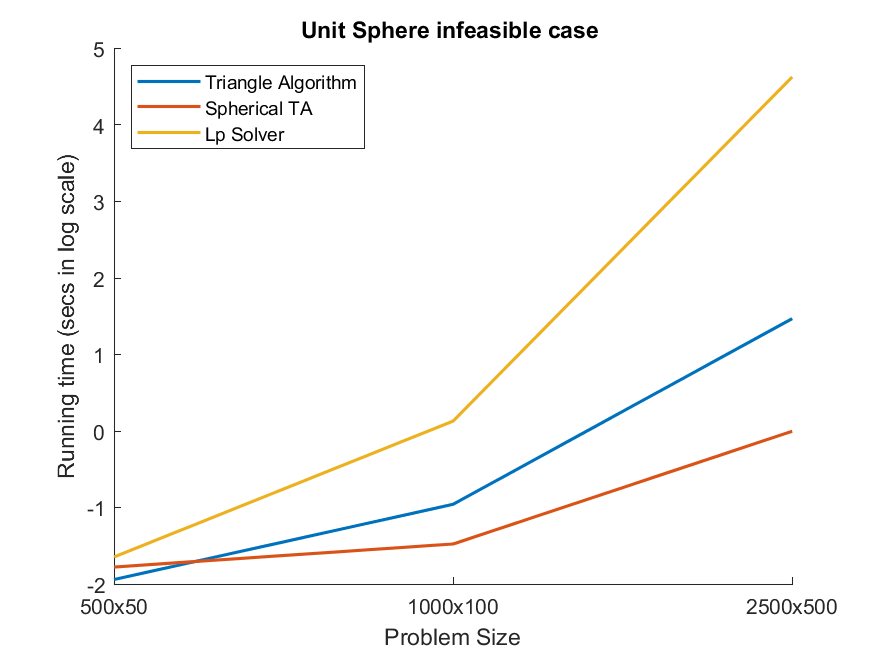}
		}
		\caption{Running time of different algorithms on Strict LP feasibility problem (Gaussian)}
		\label{Tab_str_LP_fea_Sp}
	\end{figure}

	\subsection{The Irredundancy Problem} \label{Exp_irr}
	
	\subsubsection{Finding all vertices}\label{Exp_fav}
	Here we apply the AVTA, AVTA$+$ and  Quickhull ~\cite{barber1996quickhull} to solve the irredundancy problem.  We compare the efficiency of the three algorithms during the execution of which we control different parameters: 1) the dimension of the problem;  2) the number of points in $S$; 3) the fraction of redundant points in $S$ i.e., fraction of points inside $conv(S)$; and 4) the distributions used to generate the vertices. We generate vertices according to a Gaussian distribution $\mathcal{N}(0,1) ^m$ or uniformly randomly from a unit sphere. Having generated the set of vertices,  redundant points are generated as convex combination of the vertices. The size of the problems varies from $m=5,n=50 $ to $ m=200,n=1000$ and the fraction of non-vertex points varies from $0 \sim 50 \% $. The running times of the three
	algorithms are shown in Figures \ref{Tab_Irr_U} and \ref{Tab_Irr_G}. While Quickhull performs better in small size problems, especially for low dimension, it fails to output the vertices with dimension $m > 10$ in any reasonable time which is due to its exponential dependence on dimension $m$ in the complexity. The AVTA and the AVTA$+$ demonstrates significantly better efficiency in large size problems since finding a vertex only takes linear time in $m,n$.

	\begin{figure}[!h]
		\centering
		\subfloat[]{
			\centering
			\includegraphics[width=0.451\textwidth]{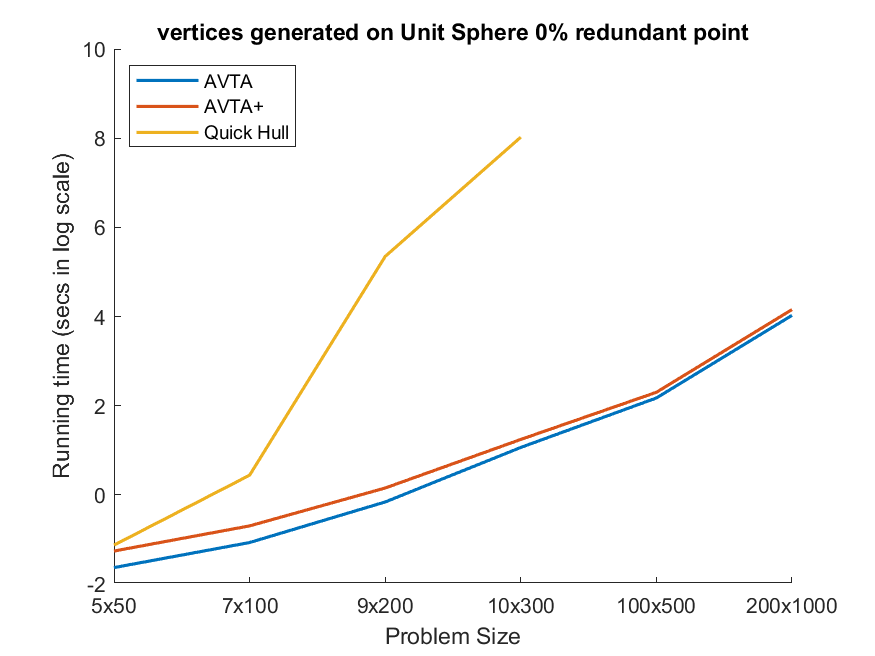}
e		}
		\subfloat[]{
			\centering
			\includegraphics[width=0.451\textwidth]{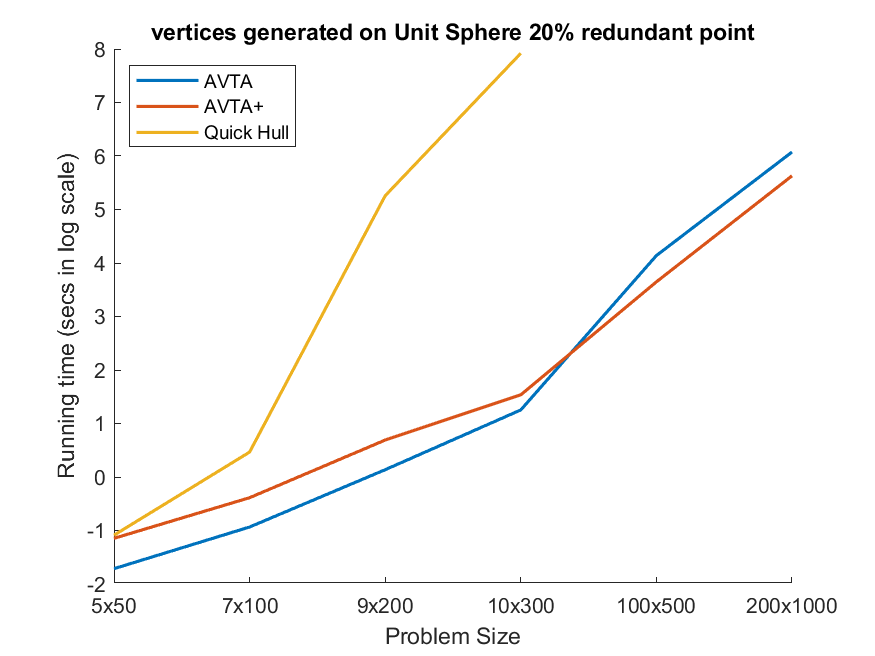}
		}\\
		\subfloat[]{
			\centering
			\includegraphics[width=0.451\textwidth]{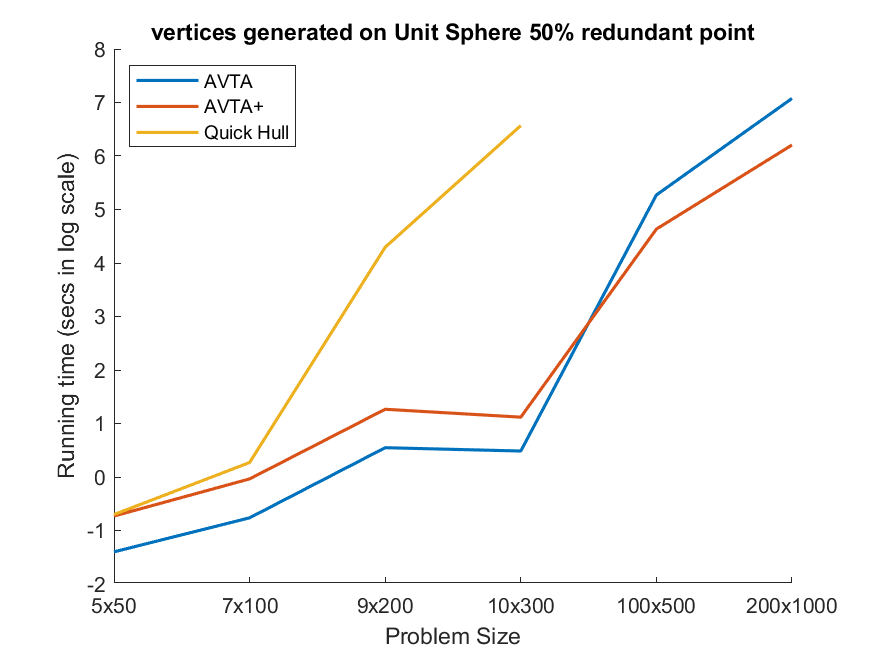}
		}
		\caption{Running time of different algorithms on Irredundancy problem (Unit Sphere)}
		\label{Tab_Irr_U}
	\end{figure}
	
	\begin{figure}[!h]
		\centering
		\subfloat[]{
			\centering
			\includegraphics[width=0.451\textwidth]{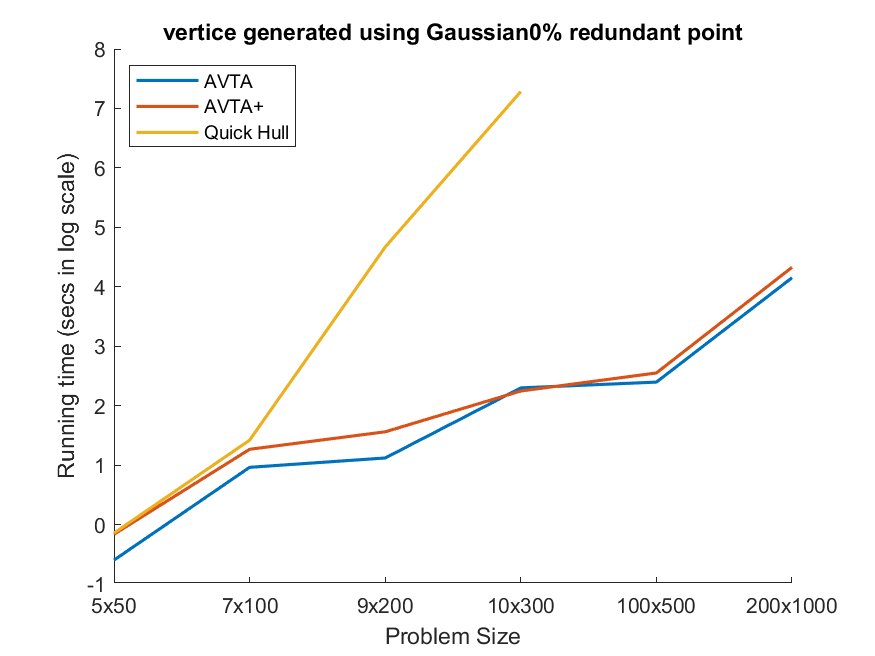}
		}
		\subfloat[]{
			\centering
			\includegraphics[width=0.451\textwidth]{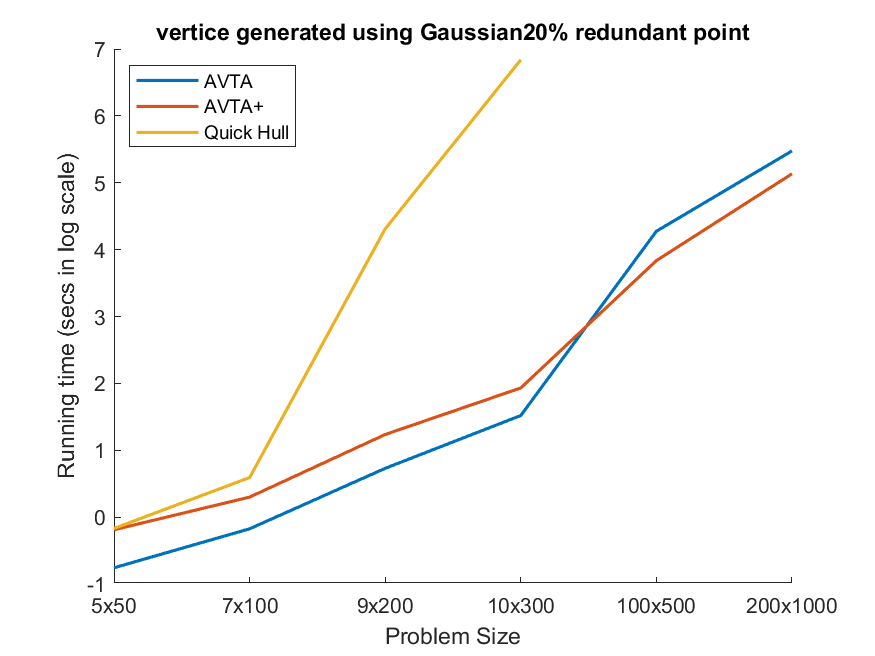}
		}\\
		\subfloat[]{
			\centering
			\includegraphics[width=0.451\textwidth]{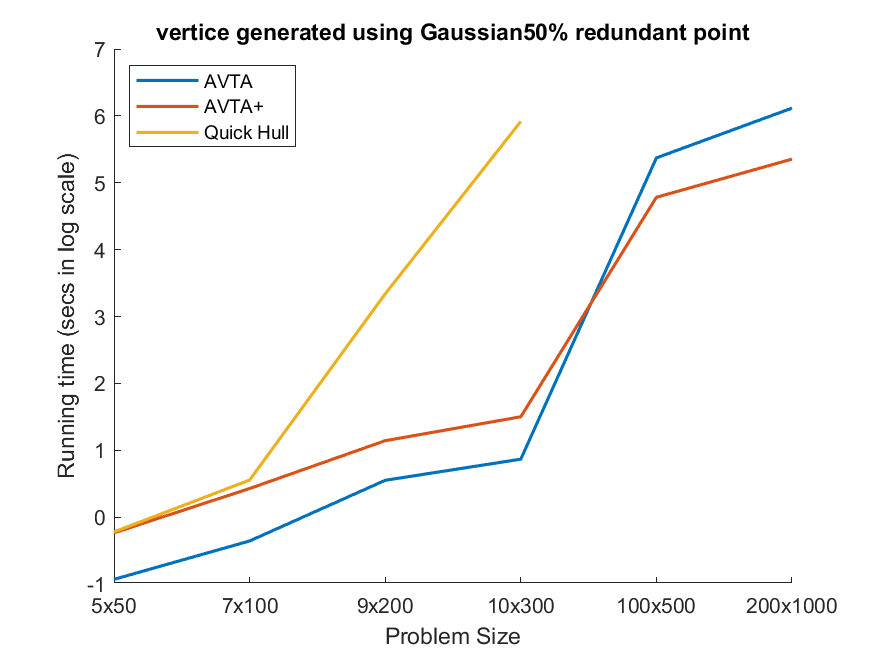}
		}
		\caption{Running time of different algorithms on Irredundancy problem (Gaussian)}
		\label{Tab_Irr_G}
	\end{figure}

	\subsubsection{Minimum Volume Enclosing Ellipsoid}\label{Exp_mvee}
	Here we show that the AVTA and the AVTA$+$ can handle large scale overcomplete data in the MVEE problem. \noindent In our experiments, vertices of the convex hull are generated from a Gaussian distribution. We set the number of vertices $K=500$. Having generated the vertices, the 'redundant' points $d_j$, where $d_j \in conv(S), j=1,...,n-K$, are generated using a random convex combination $d_j=\sum_{i=1}^{K} \alpha_i v_i$. The  $\alpha_i$'s are scaled so that $\sum_{i=1}^{K} \alpha_i=1$. The algorithm \emph{AVTA$+$ and MVEE} is implemented as follows: First run AVTA$+$ on $S$ to find all vertices $\widehat{S} \subset S$, then run MVEE on $\widehat{S}$. The \emph{AVTA and MVEE} is implemented in a similar manner. The value of epsilon in Figure \ref{Tab_MVEE} is the precision parameter for solving the MVEE problem using the \textit{MinVolEllipse} function.
	The running times of the three algorithms are presented in Figure \ref{Tab_MVEE}. The results in Figure \ref{Tab_MVEE} demonstrate  that the AVTA and AVTA$+$ are an efficient pre-processing steps for data reduction, especially when the number of redundant points dominates the dataset: $n>>K$. Indeed the AVTA can reduce the size of dataset from $n$ to $K$ thus the downstream task has much smaller scale problem to solve.

	\begin{figure}[t]
		\centering
		\subfloat[]{
			\centering
			\includegraphics[width=0.451\textwidth]{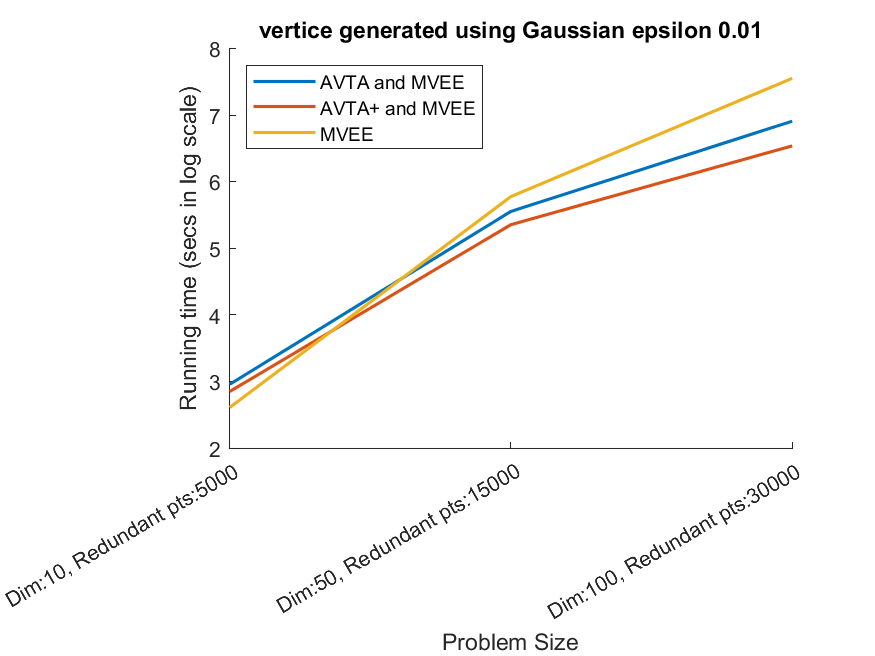}
		}
		\subfloat[]{
			\centering
			\includegraphics[width=0.451\textwidth]{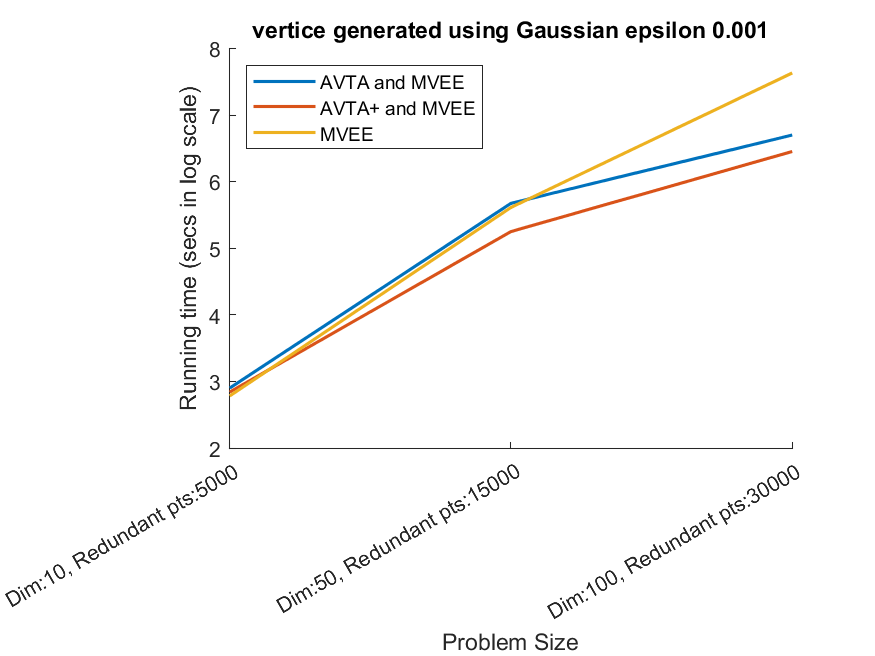}
		}\\
		\subfloat[]{
			\centering
			\includegraphics[width=0.451\textwidth]{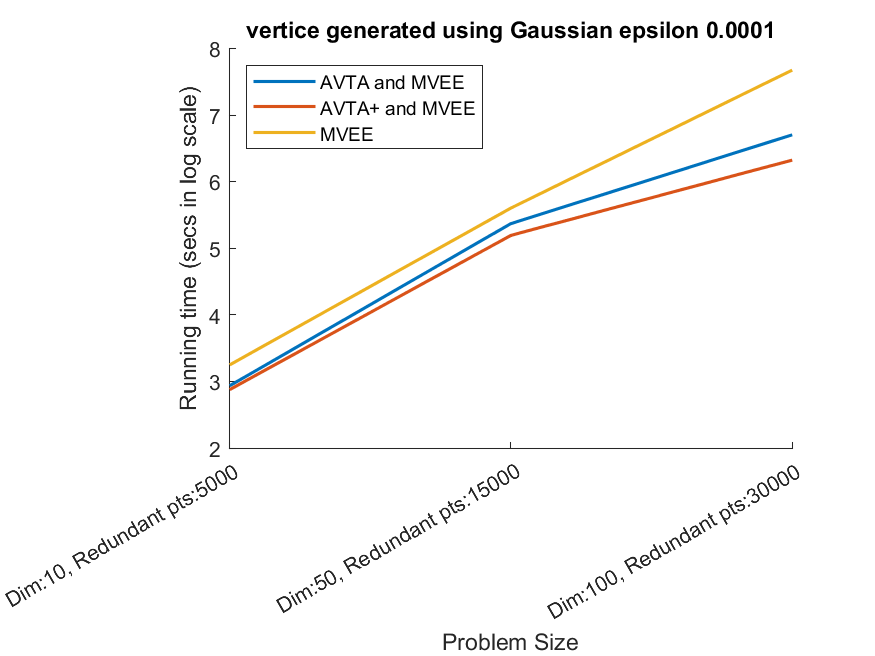}
		}
		\caption{Running time on Minimum Volume Enclosing Ellipsoid problem}
		\label{Tab_MVEE}
	\end{figure}

	\section{Concluding Remarks}
	
	In this article we considered CHM, a fundamental problem in diverse fields. We considered the special case of CHM,  Spherical-CHM, which tests if the origin lies in the convex hull of  $n$ points on the unit sphere.  This canonical formulation has important features that can be exploited algorithmically.   We first showed that both in the sense of  exact and approximate solutions, Spherical-CHM is equivalent to CHM. We then provide a variant of of the TA, called the Spherical-TA which first converts a CHM into Spherical-CHM. On the one hand, we report a novel complexity analysis for the TA to prove that under a verifiable assumption at each iteration called the $\varepsilon$-property, the  number of iterations of  TA improved to  $O(1/\varepsilon)$. On the other hand, we applied the Spherical-TA to solve a set of distinct problems. 	
Our empirical results demonstrated that the TA and the Spherical-TA achieves impressive performance in solving problems that include, CHM, LP Feasibility, Strict LP Feasibility and  the Irredundancy Problem. In particular, we applied the irredundancy for reducing data in the large scale MVEE problems.

The TA and the Spherical-TA can be used as fast membership query oracles in high dimensional problems.  Our computational results strongly support the TA and the  Spherical-TA as effective tools in areas such as Linear Programming, Computational Geometry, and Machine Learning.  Our algorithms are implemented in MATLAB and available to the readers.

%

	%
	%
	%
	%
	%
	%
	%
	%
	%
	%
	%
	%
	%
	
	\bibliographystyle{plain}
	\bibliography{acmart}
\end{document}